\documentclass[journal,twocolumn,10pt]{IEEEtran}
\usepackage{amsmath,graphicx,mathrsfs,amsfonts}

\newtheorem{theorem}{Theorem}

\newtheorem{lemma}{Lemma}

\newcounter{excounter}
\newenvironment{ex}[1][] 
{
	\refstepcounter{excounter}
	\par{\em Example \theexcounter\: (#1).}
}
{
    \hfill $\triangle$
}

\def\given{\:|\:}
\def\L{\mathsf{L}}
\def\H{\mathsf{H}}
\def\U{\mathsf{U}}
\def\B{\mathsf{B}}

\newcommand{\ciid}{C_{\mathrm{iid}}}
\def\Ca{Ca$^{2+}$}

\newcommand{\binent}{\mathscr{H}}

\newcommand{\C}{\mathsf{C}}
\newcommand{\N}{\mathsf{N}}
\renewcommand{\O}{\mathsf{O}}

\newcommand{\xlevel}{\mathsf{x}}
\newcommand{\xl}{\xlevel_\L}
\newcommand{\xh}{\xlevel_\H}
\usepackage{color}

\newcommand{\pt}[1]{{\color{black} #1}}

\newcommand{\awe}[1]{{\color{black} #1}}
\newcommand{\rev}[1]{{\color{black} #1}}

\def\Pr{{\mathrm{Pr}}}

\title{\pt{The Channel Capacity of Channelrhodopsin\\ and Other Intensity-Driven Signal\\ Transduction Receptors}}
\author{Andrew W. Eckford, \IEEEmembership{Senior Member, IEEE}, and Peter J. Thomas
\thanks{Material in this paper was presented in part at the 2016 IEEE International Conference on Acoustics, Speech, and Signal Processing (ICASSP), Shanghai, China.}%
\thanks{This work was supported in part by a grant from the Natural Sciences and Engineering Research Council, and by grants from the National Science Foundation (DMS-1413770 and DEB-1654989). This research has also been supported in part by the Mathematical Biosciences Institute and the National Science Foundation under grant DMS 1440386.
}%
\thanks{Andrew W. Eckford is with the Department of Electrical Engineering and Computer Science, York University, 4700 Keele Street, Toronto, Ontario, Canada M3J 1P3. E-mail: aeckford@yorku.ca}
\thanks{Peter J. Thomas is with the Department of Mathematics, Applied Mathematics, and Statistics, and the Department of Biology, and the Department of Electrical Engineering and Computer Science, Case Western Reserve University, Cleveland, Ohio, USA 44106. E-mail: pjthomas@case.edu}
}
\begin{document}
%
\maketitle
\begin{abstract}
Biological systems transduce signals from their surroundings through a myriad of pathways. In this paper, we describe signal transduction as a communication system: the signal transduction receptor acts as the receiver in this system, and can be modeled as a finite-state Markov chain with transition rates governed by the input signal.
Using this general model, we give the mutual information under IID inputs in discrete time, and obtain the mutual information in the continuous-time limit. We show that the mutual information has a concise closed-form expression with clear physical significance. We also give a sufficient condition under which the Shannon capacity is achieved with IID inputs.
We illustrate our results with three examples: the light-gated  Channelrhodopsin-2 (ChR2) receptor; the ligand-gated nicotinic acetylcholine (ACh) receptor; and the ligand-gated calmodulin (CaM) receptor. In particular, we show that the IID capacity of the ChR2 receptor is equal to its Shannon capacity. We finally discuss how the results change if only certain properties of each state can be observed, such as whether an ion channel is open or closed.
\end{abstract}

\section{Introduction}

Living cells take in information from their surroundings through myriad \emph{signal transduction} processes.
Signal transduction takes many forms: the input signal can be carried by changes in chemical concentration, electrical potential, light intensity, mechanical forces, and temperature, \textit{inter alia}.  In many instances these \emph{extracellular} stimuli trigger \emph{intracellular} responses that can be represented as transitions among a discrete set of states \cite{hlavacek2006}.
Models of these processes are of great interest to mathematical and theoretical biologists \cite{janes2006}.

The ``transduction'' of the signal occurs through the physical effect of the input signal on the transition rates among the various states describing the receptor.  
\pt{An early} 
mathematical model of this type was the voltage-sensitive transitions among several open and closed ion channel states in Hodgkin and Huxley's  model for the conduction of sodium and potassium ions through the membranes of electrically excitable cells \cite{hodgkin1990}. 
Presently, many such models are known for signal transduction systems, such as: the detection of calcium concentration signals by the calmodulin protein \cite{faas11}, binding of the acetylcholine (ACh) neurotransmitter to its receptor protein \cite{col82}, and modulation of the channel opening transition by light intensity in the channelrhodopsin (ChR) protein \cite{nag03}.  In each of these examples the channel may be modeled as a weighted, directed graph, in which the vertices represent the discrete channel states, and the weighted edges represent \textit{per capita} transition rates, some of which can be modulated by the input signals.

Mutual information, and Shannon capacity, arise in a variety of biological contexts. For example, mutual information may predict the differential growth rates of organisms learning about their environment \cite{tkacik2016}, based on the Kelly criterion \cite{kelly1956}. For biological communication systems, achieving a distortion criterion (expressed as mutual information) need not require complicated signal processing techniques; see \cite[Example 2]{gastpar2003}. Moreover, the free energy cost of molecular communication (such as in signal transduction) has a mathematical form similar to mutual information \cite{eckford2018}, leading to thermodynamic bounds on capacity per unit energy cost (cf. \cite{verdu1990}).

Stochastic modeling of signal transduction as a communication channel has considered the chemical reactions in terms of Markov chains \cite{thomas2003} and in terms of the ``noise'' inherent in the binding process \cite{pierobon2011}. For simplified two-state Markov models, Shannon capacity of signal transduction has been calculated for slowly changing inputs \cite{ein11} and for populations of communicating bacteria \cite{aminian2015}. Our own previous work has investigated the capacity of signal transduction: in \cite{ThomasEckford2016}, we obtained the Shannon capacity of two-state Markov signal transduction under arbitrary inputs, and showed that the capacity for multiple independent receptors has the same form \cite{thomas2016}. Related channel models have been studied in the information-theoretic literature, such as the unit output memory channel \cite{che05}, the ``previous output is the state'' (POST) channel \rev{\cite{AsnaniPermuterWeissman2013IEEE_ISIT,PermuterAsnaniWeissman2014IEEETransIT}}; capacity results for some channels in these classes were recently given in \cite{stavrou2016}.

The present paper \pt{focuses on} 
the mutual information and capacity of finite-state signal transduction channels. Generalizing previous results, we provide discrete-time, finite-state channel models for a wide class of signal transduction receptors, giving Channelrhodopsin-2 (ChR2), Acetylcholine (ACh), and Calmodulin (CaM) as specific examples. 
We also provide an explicit formula for the mutual information of this \pt{class of models} 
under independent, identically distributed (IID) inputs \rev{(Theorem 1)}. 
Subsequently, we consider 
the continuous time limit
as the interval between the discrete-time instants goes to zero, and find a simple closed-form expression for the mutual information \rev{(Theorem 2)}, with a natural physical interpretation. 
We further give conditions under which our formula gives the Shannon capacity of the channel, namely that there is exactly one transition in the Markov chain that is sensitive to the channel input \rev{(Theorem 3), and we use this result to (numerically) find the Shannon capacity of ChR.}

The remainder of the paper is organized as follows: in \rev{Section} II, we give a generalized model for discrete-time, finite-state signal transduction systems; in \rev{Section} III, we discuss signal transduction as a communication system, deriving expressions for the mutual information and giving our main results; and in \rev{Section} IV, we discuss the biological significance of the results, \pt{as well as the limitations of our analysis}.

\section{Model}
\label{sec:Model}

\subsection{Physical model}


Signal transduction encompasses a wide variety of physical processes. 
For example, in a {\em ligand-gated} system, signals are transmitted using concentrations of signaling molecules, known as {\em ligands}, which bind to receptor proteins. As another example, in a {\em light-gated} system, signals are transmitted using light, where the receptor absorbs photons. Other possibilities exist, such as voltage-gated ion channels. The receptor, \pt{often} 
located on the surface of the cell, forms the receiver in the signal transduction system, and conveys (or {\em transduces}) the signal across the cell membrane; the receptor is the focus of our analysis.

Signal transduction receptors share a mathematical model: they can be viewed as finite-state, intensity-modulated Markov chains, in which the transition rates between certain pairs of states are sensitive to the input (though other transitions may be independent of the input).
Our main examples in this paper focus on ligand- and light-gated receptors. For example, in a ligand-gated system, the binding of the ligand results in a change in the receptor, which then produces {\em second messengers} (normally a different species than the ligand) to convey the message to the cell interior. In a light-gated system, the incident photon causes a similar change in the receptor, which may open to allow an ion current to pass to the interior of the cell. 
In either case, there may be a relaxation process which returns the receptor to the ``ready'' state, and this process may be independent of the signal; or other processes that are either sensitive to or independent of the signal, depending on the purpose of the receptor.

In the next two sections, we describe the Markov chain model for receptors, both in continuous and in discrete time.
Although we focus on ligand- and light-gated receptors, we emphasize that our framework is general enough to include other kinds of receptors.



\subsection{Continuous time: Master equation kinetics}
\label{sec:MasterEquation}

Receptors are finite-state Markov chains. For a receptor with $k$ discrete states, 
there exists a $k$-dimensional vector of state occupancy probabilities $\mathbf{p}(t)$,
given by
\begin{equation}
	\mathbf{p}(t) = \left[ p_1(t), \: p_2(t), \: \ldots, \:p_{k}(t) \right] ,
\end{equation}
where $p_i(t)$ represents the probability of a given receptor occupying state $i$ at time $t$. The environmental conditions at the receptor, such as light level or ligand concentration, are known as the {\em input} $x(t)$.

The chemical kinetics of the receptor are captured by a differential equation known
as the {\em master equation} \cite{Gardiner2004}. Let $Q = [q_{ij}\rev{(x)}]$ represent a $k \times k$ matrix of \rev{\textit{per capita} transition rates,}
where $q_{ij}\rev{(x)}$ represents the instantaneous rate at which receptors starting in state $i$
enter state $j$. 
It is helpful to visualize the matrix $Q$ using a graph:
\begin{itemize}
	\item There are $k$ vertices, representing the states; and  
    \item A directed edge is drawn from vertex $i$ to $j$ if and only if $q_{ij}\rev{(x)} > 0$ \rev{for some  $x$}. 
\end{itemize}
Changing from one state to another is called a {\em transition}, so the graph corresponding to $Q$ depicts the possible transitions. A transition $i \rightarrow j$ may be {\em sensitive}, i.e. $q_{ij}$ \rev{varies as} a function of the input $x(t)$, or {\em insensitive}, $q_{ij}$ is constant with respect to $x(t)$.

Using $Q$, the master equation is given by
\begin{equation}
	\label{eqn:MasterEquation}
	\frac{d\mathbf{p}(t)}{dt} = \mathbf{p}(t)Q\rev{(x(t))} .
\end{equation}
%
%
We use the notation from \cite{GroffDeRemigioSmith2009chapter}:
%
\begin{itemize}
	\item States take a compound label, consisting of a state property and a state number. The state number is unique to each state, but the state property may be shared by multiple states. For example, in each state the receptor's ion channel might be either open $\O$ or closed $\C$; the state label $\C_1$ means that in state 1 the channel is closed, and $\O_2$ means that in state 2 the channel is open. \awe{In this paper we use the state {\em number} rather than the state {\em property}. (Since we show that the state numbers form a Markov chain, in general the state properties form a hidden Markov chain; we discuss this further in Section IV.)}
	\item \rev{We assume that} rates which are sensitive to the input are directly proportional to the input $x(t)$. For example, 
	$q_{12}x(t)$ is the transition rate from $1 \rightarrow 2$, 
	which is sensitive, while  
	$q_{31}$ is the transition rate from $3 \rightarrow 1$, which is insensitive.
	\item The $i$th diagonal element \rev{of $Q$} is written $R_i$, and is set so that the $i$th row sums to zero (so, if $x(t)$ appears in the $i$th row, $R_i$ may depend on $x(t)$).
\end{itemize}
Taking sensitive rates to be proportional to the signal $x(t)$ is a key modeling assumption; it is satisfied for the examples we consider, but there exist systems in which the signal acts nonlinearly on the rate.

The following three examples illustrate the use of our notation, and give practical examples of receptors along with their transition graphs and rate constants.

\begin{ex}[Channelrhodopsin-2] 
The Channelrhodopsin-2 (ChR2) receptor is a light-gated ion channel.
The receptor has three states, named Closed ($\mathsf{C}_1$), Open ($\mathsf{O}_2$), and Desensitized ($\mathsf{C}_3$).
The channel-open ($\mathsf{O}$) state $\mathsf{O}_2$ is the only state in which the ion channel is open, passing an ion current.
The channel-closed ($\mathsf{C}$) states, $\mathsf{C}_1$ and $\mathsf{C}_3$, are distinct in that the receptor is light-sensitive in
state $\mathsf{C}_1$, and insensitive in state $\mathsf{C}_3$ \cite{nag03}. The rate matrix for ChR2 is
\begin{equation}
	\label{eqn:ChR2-rate-matrix}
	Q = \left[
		\begin{array}{ccc}
			R_1 & q_{12}x(t) & 0 \\
			0 & R_2 & q_{23} \\
			q_{31} & 0 & R_3
		\end{array}
	\right] .
\end{equation}
where $x(t) \in [0,1]$ is the relative intensity.
To keep the row sums equal to zero, we set $R_1 = - q_{12}x(t)$,
$R_2 = - q_{23}$, and
$R_3 = - q_{31}$.
Fig.~\ref{fig:ChR2} shows
state labels and allowed state transitions.
\begin{figure}
	\begin{center}
	\includegraphics[width=2.5in]{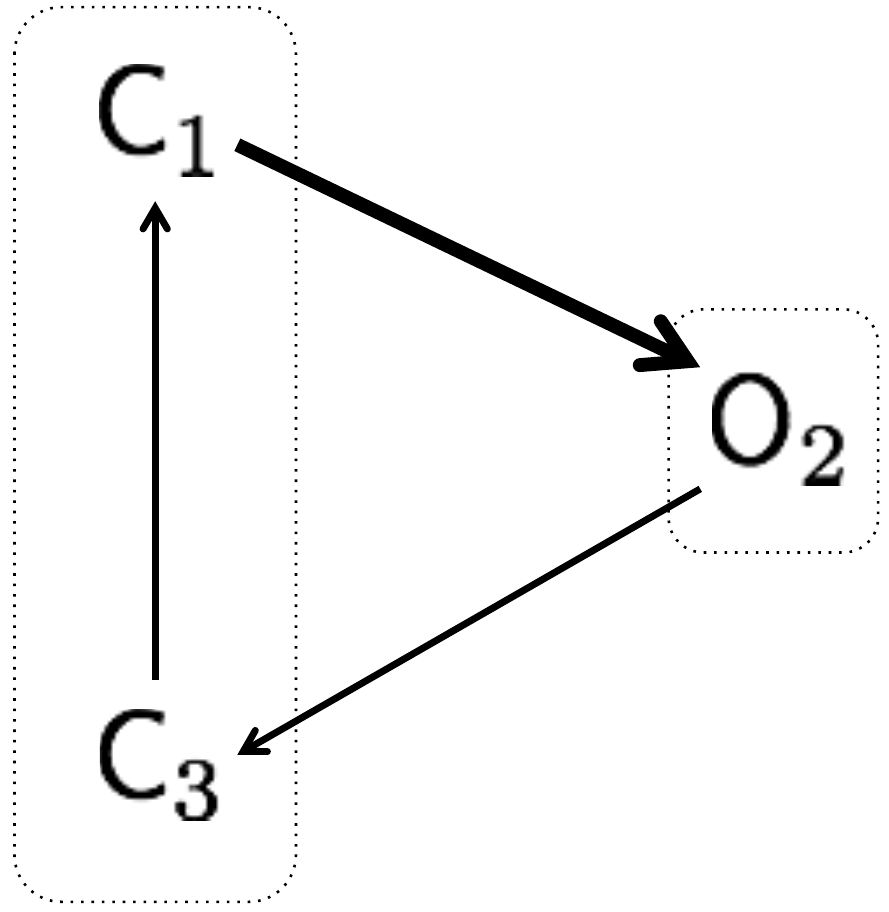}
	\end{center}
	\caption{\label{fig:ChR2} Depiction of allowed state transitions for ChR2. Sensitive 
	transitions are depicted with {\bfseries bold} arrows. 
	States are labelled by their ion channel state: $\{\mathsf{C},\mathsf{O}\}$ for closed and open, respectively; state number is in subscript. Dashed lines surround all states in either the closed or open state. Transition rates, listed in Table \ref{tab:ChR2parameters}, correspond to the vertices associated with each directed edge: for example,
	the rate from state $\mathsf{O}_2$ to state $\mathsf{C}_3$ is $q_{23}$.}
\end{figure}
Parameter values from the literature are given in Table \ref{tab:ChR2parameters}. 
%
\begin{table}[h!]\begin{center}
	\begin{tabular}{c|c|c} \hline
		Parameter & from \cite{nag03} & Units \\ \hline
		$q_{12}x(t)$ & $(5 \times 10^3) x(t)$ &  s$^{-1}$   \\ \hline
		$q_{23}$ & 50 &  s$^{-1}$   \\ \hline
		$q_{31}$ & 17 &  s$^{-1}$ \\ \hline
	\end{tabular}
	\ \\
    \ \\
	\caption{\label{tab:ChR2parameters} Rate parameters for ChR2, adapted from \cite{nag03}, where $x(t) \in [0,1]$ represents the relative light intensity.}
	\end{center}
\end{table}
\end{ex}

\begin{ex}[Acetylcholine]
The Acetylcholine (ACh) receptor is a ligand-gated ion channel. \pt{Following \cite{col82}, we model the receptor as a conditional Markov process on} five states, with rate matrix
\begin{equation}
	\label{eqn:AChRateMatrix}
	Q =
	\left[
		\begin{array}{ccccc}
			R_1 & q_{12}x(t) & 0 & q_{14} & 0 \\
			q_{21} & R_2 & q_{23} & 0 & 0 \\
			0 & q_{32} & R_3 & q_{34} & 0 \\
			q_{41} & 0 & q_{43}x(t) & R_4 & q_{45} \\
			0 & 0 & 0 & q_{54}x(t) & R_5 
		\end{array}
	\right] .
\end{equation}
%
%

There are three sensitive transitions: $\rev{q_{12}}x(t)$, $\rev{q_{43}}x(t)$, and $\rev{q_{54}}x(t)$, which 
are proportional to ligand concentration $x(t)$. 
For the purposes of our analysis, we use a range of $x(t) \in [10^{-7},10^{-5}]$.
Fig.~\ref{fig:ACh} shows the
allowed state transitions. 

The states in ACh correspond to the binding of a ligand to one of two binding sites on the receptor. In state $\mathsf{C}_5$, neither site is occupied; in states $\mathsf{C}_4$ and $\mathsf{O}_1$, one site is occupied; and in states $\mathsf{C}_3$ and $\mathsf{O}_2$, both sites are occupied.  

Table \ref{tab:AChparameters} gives
parameter values; the concentration of ACh, $x(t)$, is measured in mol/$\ell$.

The same state-naming convention is used in the figure as with ChR2: states with
an open ion channel are $\mathsf{O}_{1}$ and $\mathsf{O}_2$; states 
with a closed ion channel are $\mathsf{C}_3$, $\mathsf{C}_4$, and $\mathsf{C}_5$.
\begin{table}\begin{center}
	\begin{tabular}{c|c|c|c} \hline
		Parameter & Name in \cite{col82} & Value/range & Units \\ \hline
		$q_{12}x(t)$ & $k_{+2}x$ & $(5 \times 10^8) x(t)$ & s$^{-1}$ \\ \hline
		$q_{14}$ & $\alpha_1$ & $3 \times 10^3$ & s$^{-1}$ \\ \hline
		$q_{21}$ & $2 k_{-2}^*$ & $ 0.66 $ & s$^{-1}$ \\ \hline
		$q_{23}$ & $\alpha_2$ & $5 \times 10^2$ & s$^{-1}$ \\ \hline
		$q_{32}$ & $\beta_2$ & $1.5 \times 10^4$ & s$^{-1}$ \\ \hline
		$q_{34}$ & $2 k_{-2}$ & $ 4 \times 10^3$ & s$^{-1}$ \\ \hline 
		$q_{41}$ & $\beta_1$ & 15 & s$^{-1}$ \\ \hline
		$q_{43}x(t)$ & $k_{+2}x$ & $(5 \times 10^8) x(t)$ & s$^{-1}$  \\ \hline
		$q_{45}$ & $k_{-1}$ & $ 2 \times 10^3$ & s$^{-1}$ \\ \hline
		$q_{54}x(t)$ & $2 k_{+1} x$ & $(1 \times 10^8) x(t)$ & s$^{-1}$  \\ \hline
\end{tabular}
\end{center}
	\ \\
	\caption{\label{tab:AChparameters}Rate parameters for ACh, adapted from \cite{col82}, where $x(t)$ represents the molar concentration of ACh in mol/$\ell$. Here we use a range of $x(t) \in [10^{-7},10^{-5}]$.}
\end{table}
\begin{figure}
	\begin{center}
	\includegraphics[width=3in]{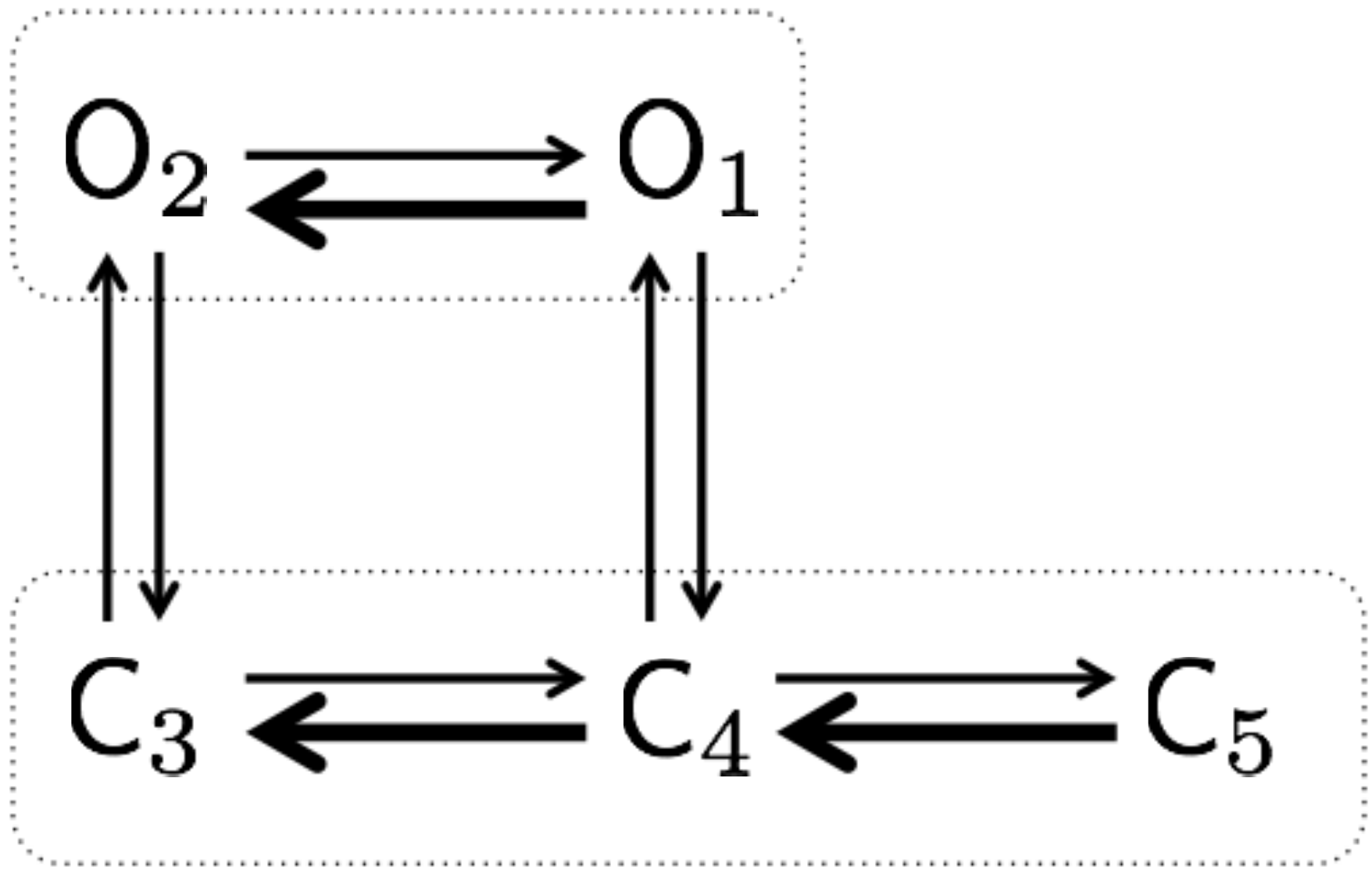}
	\end{center}
	\caption{\label{fig:ACh} Depiction of allowed state transitions for ACh. Sensitive 
	transitions are depicted with {\bfseries bold} arrows. 
	States are labelled by their ion channel state: $\{\mathsf{C},\mathsf{O}\}$ for closed and open, respectively; state number is in subscript. Dashed lines surround all states in either the closed or open state. Transition rates, listed in Table \ref{tab:AChparameters}, correspond to the vertices associated with each directed edge: for example,
	the rate from state $\mathsf{O}_2$ to state $\mathsf{C}_3$ is $q_{23}$.}
\end{figure}
\end{ex}

\begin{ex}[Calmodulin]
The Calmodulin (CaM) receptor is a ligand-gated receptor.
The CaM receptor consists of \pt{four} binding sites, \pt{two on the C-terminus of the CaM protein and two on the N-terminus \cite{ChinMeans2000TrendsCellBiol,DeMariaEtAlYue2001Nature,KellerFranksBartolSejnowski2008PLoSOne}.}  
Each \pt{end of the protein} can bind 0, 1, or 2 \pt{calcium ions, leading to} nine possible states.
For CaM, rather than an ion channel, it is important whether the $\C$ or $\N$ end of the receptor is completely bound (i.e., has both binding sites occupied by ligands). This property is represented by four symbols: $\emptyset$ if neither end is completely bound; $\C$ if the $\C$ end is completely bound; $\N$ if the $\N$ end is completely bound; and $\N\C$ if both ends are completely bound.

\begin{figure*}[!t]
\normalsize
%
\begin{align}
	\label{eqn:CaMQ}
	Q &= \left[
		\begin{array}{ccccccccc}
			R_0 & q_{01}x(t) & 0 & q_{03}x(t) & 0 & 0 & 0 & 0 & 0 \\
			q_{10} & R_1 & q_{12}x(t) & 0 &q_{14}x(t) & 0 & 0 & 0 & 0\\
			0 & q_{21} & R_2 & 0 & 0 & q_{25}x(t) & 0 & 0 & 0 \\
			q_{30} & 0 & 0 & R_3 & q_{34}x(t) & 0 & q_{36}x(t) & 0 & 0 \\
			0 & q_{41} & 0 & q_{43} & R_4 & q_{45}x(t) & 0 & q_{47}x(t) & 0 \\
			0 & 0 & q_{52} & 0 & q_{54} & R_5 & 0 & 0 & q_{58}x(t) \\
			0 & 0 & 0 & q_{63} & 0 & 0 & R_6 & q_{67}x(t) &0 \\
			0 & 0 & 0 & 0 & q_{74} & 0 & q_{76} & R_7 & q_{78}x(t) \\
			0 & 0 & 0 & 0 & 0 & q_{85} & 0 & q_{87} & R_8
		\end{array}
	\right] 
\end{align}
%
\hrulefill
\vspace*{4pt}
\end{figure*}

State configuration and allowed transitions are depicted in Figure \ref{fig:CaM}. The rate matrix is given in (\ref{eqn:CaMQ}),
%
with values given in Table \ref{tab:CaMparameters}, and where the molar concentration of calcium is $x(t) \in [10^{-7},10^{-6}]$.
\begin{table}\begin{center}
	\begin{tabular}{c|c|c|c} \hline
		Parameter & Name in \cite{faas11} & Value/range & Units \\ \hline
		$q_{01}x(t)$, $q_{34}x(t)$, $q_{67}x(t)$ 
			& $k_{\mathrm{on (T), N}} $ 
			&  $(7.7 \times 10^8) x(t)$ & s$^{-1}$  \\ \hline
		$q_{10}$, $q_{43}$,  $q_{76}$ 
			& $k_{\mathrm{off (T), N}}$ 
			&  $1.6 \times 10^5$ & s$^{-1}$  \\ \hline
		$q_{12}x(t)$, $q_{45}x(t)$, $q_{78}x(t)$ 
			& $k_{\mathrm{on (R), N}}$ 
			&  $(3.2 \times 10^{10}) x(t)$ & s$^{-1}$  \\ \hline
		$q_{21}$,  $q_{54}$, $q_{87}$ 
			& $k_{\mathrm{off (R), N}}$ 
			& $2.2 \times 10^4$ & s$^{-1}$  \\ \hline
		$q_{03}x(t)$, $q_{14}x(t)$, $q_{25}x(t)$
			& $k_{\mathrm{on (T), C}}$ 
			& $(8.4 \times 10^7) x(t)$ & s$^{-1}$ \\ \hline
		$q_{30}$, $q_{41}$, $q_{52}$ 
			& $k_{\mathrm{off (T), C}}$ 
			& $2.6 \times 10^3$ & s$^{-1}$  \\ \hline
		$q_{36}x(t)$, $q_{47}x(t)$, $q_{58}x(t)$ 
			& $k_{\mathrm{on (R), C}}$ 
			& $(2.5 \times 10^7) x(t)$ & s$^{-1}$  \\ \hline
		$q_{63}$, $q_{74}$, $q_{85}$
			& $k_{\mathrm{off (R), C}}$ 
			& 6.5 & s$^{-1}$  \\ \hline
	\end{tabular}
	\end{center}
\ \\
	\caption{\label{tab:CaMparameters}Rate parameters for CaM, adapted from \cite{faas11}, where $x(t) \in [10^{-7},10^{-6}]$ represents the molar concentration of calcium in mol/$\ell$.}
\end{table}
\begin{figure}
	\begin{center}
	\includegraphics[width=3.25in]{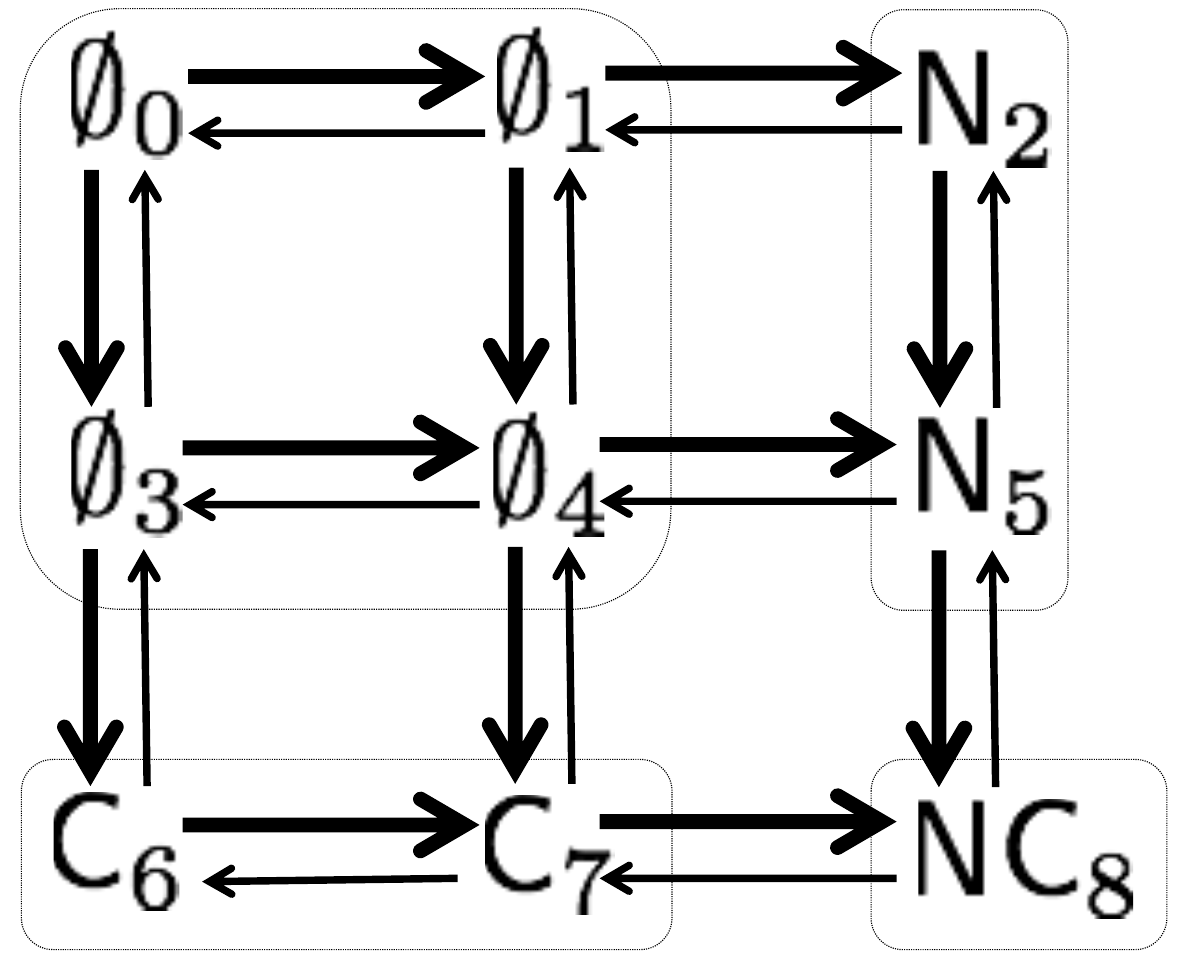}
	\end{center}
	\caption{\label{fig:CaM} Depiction of allowed state transitions for CaM. Sensitive 
	transitions are depicted with {\bfseries bold} arrows. 
	States are labelled by the status of the $\C$ or $\N$ end of the receptor: $\emptyset$ if neither end is completely bound; $\C$ if the $\C$ end is completely bound; $\N$ if the $\N$ end is completely bound; and $\N\C$ if both ends are completely bound. Transition rates, listed in Table \ref{tab:CaMparameters}, correspond to the vertices associated with each directed edge.}
\end{figure}
\end{ex}


For each of the preceding examples, the rate constants depend on environmental conditions, and
thus can be reported differently in different sources 
(see, e.g., \cite{lin09} for different rate constants for ChR2).

\subsection{From the master equation to discrete-time Markov chains}

\rev{The continuous-time master equation for the receptor dynamics \eqref{eqn:MasterEquation} describes the evolution of a conditional probability $\mathbf{p}(t)\equiv  E[Y(t)\given\mathcal{F}_X(t)],$ where $Y(t)$ is the continuous time, discrete state \textit{c\`{a}dl\`{a}g}
%
%
process giving the channel state, $\mathcal{F}_X(t)$ is the filtration generated by the input process $X(t)$, and $E[\cdot\given\cdot]$ is  conditional expectation \cite{grimmett2001probability}.
Establishing the appropriate ensemble of input processes and analyzing mutual information and capacity involve technical issues that do not shed light on the nature of biological signal transduction.
Therefore we do not undertake a rigorous analysis of the continuous-time communications channels described by \eqref{eqn:MasterEquation}  in this paper.  
Rather, we introduce a discrete-time, discrete-state channel, motivated by the continuous-time channel, which can be rigorously analyzed, and study its properties both with a fixed timestep $\Delta t$, and later in the limit $\Delta t\to 0$.   
The discrete-time Markov chain model allows us to rely on capacity results for discrete-time Markov channels.}


\rev{We obtain a discrete-time approximation to} the master equation by writing
\begin{align}
	\label{eqn:Markov-1}
	\frac{d\mathbf{p}(t)}{dt} = \mathbf{p}(t) Q  
    = \frac{\mathbf{p}(t + \Delta t) - \mathbf{p}(t)}{\Delta t} +o(\Delta t),\text{ as }\Delta t\to 0,
\end{align}
\rev{
where we simplify the notation by writing $Q(x(t))$ as simply $Q$.
Manipulating the middle and right expression in (\ref{eqn:Markov-1}) \pt{gives} }
\rev{\begin{align}
	\mathbf{p}(t + \Delta t) &= \Delta t \,\mathbf{p}(t) Q + \mathbf{p}(t) +o(\Delta t)\\
	&= \Delta t\, \mathbf{p}(t) Q + \mathbf{p}(t) I+o(\Delta t) \\
	&= \mathbf{p}(t) \left( I + \Delta t Q \right)+o(\Delta t),\text{ as }\Delta t\to 0,
\end{align}}
where $I$ is the identity matrix. 
\rev{In order to arrive at a discrete-time model, we introduce the approximation $\{\mathbf{p}_i\}_{i\in\mathbb{N}_+}$ satisfying} 
\begin{equation}
	\mathbf{p}_i = \mathbf{p}(i \Delta t)+o(\Delta t),\text{ as }\Delta t\to 0,
\end{equation}
\rev{and} arrive at a discrete-time approximation to (\ref{eqn:Markov-1}),
\begin{equation}
	\mathbf{p}_{i+1} = \mathbf{p}_i (I+\Delta t\, Q).
\end{equation}
Thus, we have a discrete-time Markov chain with transition probability matrix 
\begin{equation}
	\label{eqn:Markov-last}
	P = I + \Delta t Q.
\end{equation}
The matrix $P$ satisfies the conditions of a Markov chain transition
probability matrix (nonnegative, row-stochastic) as long as $\Delta t$ is small enough.  
\rev{However, note that $P$ (and $Q$) are dependent on $x(t)$, so the Markov chain is not generally time-homogeneous if $x(t)$ is known (cf. (\ref{eqn:HomogeneousP})).}



\section{Signal transduction as a communications system}

In this section we give our main results, in which we describe and analyze signal transduction as a communication system. A brief roadmap to our results is given as follows: we first define the communication system in terms of input, output, and channel; we give the mutual information of the general discrete-time model under IID inputs (Theorem 1 and equation (\ref{eqn:GeneralIID1})); we take the continuous-time limit of the mutual information rate, showing that the expression for mutual information has a simple factorization (Theorem 2 and equation (\ref{eqn:Theorem1})); we give a physical interpretation of the factorization in (\ref{eqn:Theorem1}); we give general conditions under which the Shannon capacity is satisfied by IID inputs (Theorem 3); and finally, we give an example calculation using ChR2 (Example 4).

\subsection{Communication model of receptors}

We now discuss how the receptors can be described as 
information-theoretic communication systems: that is, in terms of input, output, and conditional
input-output PMF.

{\em Input:} As discussed in Section II, the receptor is sensitive to given properties of the environment; previous examples included light intensity or ligand concentration. The receptor input $x(t)$ is the value of this property at the surface of the receptor.
The input is discretized in time: for integers $i$, the input is $x(i \Delta t)$; we will write
$x_i = x(i \Delta t)$. 
We will also discretize the amplitude,
so that for every $t$, $x_i \in \{\xlevel_1,\xlevel_2,\xlevel_3,\ldots,\xlevel_k\} =: \mathcal{X}$. 
We will assume that the $\xlevel_i$ are 
distinct and increasing; further, we assign the lowest and highest values special symbols:
\begin{align}
	\xl &:= \xlevel_1 \\
    \xh &:= \xlevel_k .
\end{align}
In Section II, we gave the concentrations or intensities over a range of values (such as $x(t) \in [0,1]$ for ChR2). Thus, we select $\xl$ and $\xh$ as the minimum and maximum values of this range, respectively.

{\em Output:} 
In this paper, the output $y(t)$ of the communication system is 
the receptor state number, given by the {\em subscript} of the state label: for example, if the state is $\C_3$, then $y(t) = 3$.
This is discretized to $y_i = y(i\Delta t)$. The discrete channel inputs and outputs form vectors: in terms of notation, we write $\mathbf{x} = [x_1,x_2,\ldots,x_n]$ and $\mathbf{y} = [y_1,y_2,\ldots,y_n]$.

{\em Conditional input-output PMF:}
From (\ref{eqn:Markov-1})--(\ref{eqn:Markov-last}), $\mathbf{y}$ forms a 
Markov chain given $\mathbf{x}$, so 
\begin{equation}
	\label{eqn:ReceptorMarkov}
	p(\mathbf{y}|\mathbf{x}) = \prod_{i=1}^n p(y_i \given x_i,y_{i-1}) ,
\end{equation}
where $p(y_i \given x_i,y_{i-1})$ is given by the appropriate entry in the matrix $P$, and where $y_0$ is null.%
\footnote{\rev{Notation: (1) We will drop subscripts if it is unambiguous to do so, i.e., normally $p(x)$ signifies $p_X(x)$; (2)} We say a variable is ``null'' if it vanishes under conditioning, i.e., if $y_0$ is null, then $p(y_1 \given x_1, y_0) = p(y_1 \given x_1)$.}
The following diagram \eqref{diag:xy} indicates the conditional dependencies:
\begin{equation}\label{diag:xy}
\begin{array}{ccccccccccc}
&X_1&&X_2&&X_3&&X_4&&X_5&\cdots\\
&\downarrow&&\downarrow&&\downarrow&&\downarrow&& \downarrow   \\
(Y_0)&\longrightarrow&Y_1&\longrightarrow&Y_2&\longrightarrow&Y_3&\longrightarrow&Y_4&\cdots 
\end{array}\end{equation}
As an example, consider ACh: suppose $y_{i-1} = 1$, $y_i = 2$, and $x_i = \xh$. Then from 
(\ref{eqn:Markov-last}) and Table \ref{tab:AChparameters}, we have
$p_{Y_i \given Y_{i-1},X_i}(2 \given 1,\xh) = \Delta t q_{12}(t) = (5 \times 10^8) \xh \Delta t$.  
From (\ref{diag:xy}) and the definition of 
$P$, $p(y_i \given y_{i-1},x_i)$ does not depend on $i$; 
that is, the channel's input-output structure is time-invariant. 

For a discrete-time Markov chain, the receptor states form a graph with vertex set $\mathcal{Y}$ and directed edges $\mathcal{E}\subset\mathcal{Y}\times\mathcal{Y}$, with pair $(y_{i-1},y_i)\in\mathcal{E}$ if $\max_{x_i\in\mathcal{X}} p(y_i\given x_i,y_{i-1}) > 0$, that is, for at least some input value there is a direct transition from $y_{i-1}$ to $y_i$. Notice that, under this definition, self-transitions are {\em included} in $\mathcal{E}$, even though (for convenience) they are not depicted in the state-transition diagrams.

We say the transition from state $y_{i-1}$ to $y_i$ is \emph{insensitive to the input}, or just \emph{insensitive}, if, for all $x_i \in \mathcal{X}$, we have $p(y_i \given x_i,y_{i-1})=p(y_i \given y_{i-1})$ 
(see Section \ref{sec:MasterEquation}).
Otherwise, the transition is \emph{sensitive}. We let $\mathcal{S}\subseteq\mathcal{E}$ denote the subset of sensitive edges.
(If state $y_{i-1} \in \mathcal{Y}$ is the origin for a sensitive transition, i.e., there is at least one $(y_{i-1},y_i \neq y_{i-1}) \in \mathcal{S}$, then the self-transition $(y_{i-1},y_i = y_{i-1})$ is normally sensitive as well, but this condition is not required for our analysis.)

For a channel with inputs $\mathbf{x}$ and  outputs $\mathbf{y}$ \rev{(both of length $n$)}
the mutual information $I(\mathbf{X};\mathbf{Y})$ gives the maximum information rate that may be 
transmitted reliably over the channel for a given input distribution. 
Mutual information is given by
\begin{align}
 	I(\mathbf{X};\mathbf{Y}) 
 	\label{eqn:MutualInformation}
 	&= \sum_{\mathbf{x},\mathbf{y}} p(\mathbf{x}) p(\mathbf{y}\given\mathbf{x})  
 		\log \frac{p(\mathbf{y}\given \mathbf{x})}{p(\mathbf{y})} ,
\end{align}
where $p(\mathbf{y}\given \mathbf{x})$ is the conditional probability mass function (PMF) of $\mathbf{Y}$.  
%

As $n \rightarrow \infty$, generally $I(\mathbf{X};\mathbf{Y}) \rightarrow \infty$ as well; 
in this case, it is more useful to calculate the mutual information rate, \rev{which we introduce in the next section.}

%

\subsection{Receptor IID capacity}

\rev{Our focus in the remainder of this paper is on IID input distributions.} Although
IID inputs \rev{may not be} realistic for chemical diffusion channels, such as for ligand-gated receptors (as concentration 
may persist for long periods of time), they can be capacity-achieving in these channels \rev{(see, e.g., \cite{ThomasEckford2016})}; moreover, IID input distributions may be physically realistic for light-gated channels.

\rev{Starting with (\ref{eqn:MutualInformation}), where $\mathbf{x}$ and $\mathbf{y}$ are both of fixed and finite length $n$, the Shannon capacity $C(n)$ is found by maximizing $I(\mathbf{X};\mathbf{Y})$ with respect to the input distribution $p(\mathbf{x})$, i.e.,
\begin{align}
    \label{eqn:CapacityDefinition}
    C(n) = \max_{p(\mathbf{x})} I(\mathbf{X};\mathbf{Y}) .
\end{align}
%
%
where the limit is taken over all possible length-$n$ input distributions $p(\mathbf{x})$ (not necessarily IID).}

\rev{
If the input $p(\mathbf{x})$ is restricted to the set of IID input distributions,
which is well defined for each $n$ (i.e., $p(\mathbf{x}) = \prod_{i=1}^n p(x_i)$),
then $I(\mathbf{X};\mathbf{Y})$ is also well defined for each $n$ (see (\ref{eqn:MutualInformation})). Furthermore, for each $n$
we have the IID capacity, written $\ciid(n)$:
\begin{equation}
	\label{eqn:IIDCapacityDefinition}
    \ciid(n) = \max_{p(x_i)} I(\mathbf{X};\mathbf{Y}) .
\end{equation}
where the maximum is taken over all possible settings of $p(x_i)$.
}

\rev{We can use (\ref{eqn:MutualInformation}) and (\ref{eqn:IIDCapacityDefinition}) to obtain information rates per channel use. For a given IID input distribution $p(x)$,
the IID mutual information rate is given by
\begin{align}
    \label{eqn:InfoRateDefinition}
    \mathcal{I}(X;Y) = \lim_{n \rightarrow \infty} \frac{1}{n} I(\mathbf{X};\mathbf{Y}).
\end{align}
Furthermore, the maximum IID information rate is given by
%
%
\begin{align}
    \label{eqn:ciid}
    \ciid := \lim_{n \rightarrow \infty} \frac{1}{n} \ciid(n) .
\end{align}
We derive these quantities in the remainer of the section, in which it will be clear that the limits in (\ref{eqn:InfoRateDefinition})---(\ref{eqn:ciid}) exist. 
We start by deriving $I(\mathbf{X};\mathbf{Y})$ under IID inputs, and showing how it is calculated using quantities introduced in Section II. Finally, in Theorem \ref{thm:MutualInformationRate}, we give an expression for $\mathcal{I}(X;Y)$, and show that $\mathcal{I}(X;Y) = \ciid$.
}




Recall $p(\mathbf{y}\given\mathbf{x})$ from (\ref{eqn:ReceptorMarkov}).
Under IID inputs, it can be shown (see \cite{che05,ThomasEckford2016}) that the receptor states $Y^n$ form a time-homogeneous Markov
chain, that is,
\begin{equation}
	\label{eqn:Conditional-2}
	p(\mathbf{y}) = \prod_{i=1}^n p(y_i \given y_{i-1}),
\end{equation}
where $y_0$ is again null, and where 
\begin{equation}
    \label{eqn:Conditional-2a}
	p(y_i \given y_{i-1}) = \sum_{x_i} p(y_i \given x_i, y_{i-1}) p(x) .
\end{equation}
\rev{
Furthermore, let $\bar{P}$ represent the transition probability matrix of $Y^n$. 
Recall (\ref{eqn:Markov-last}), in which $P$ was dependent on $x$; using (\ref{eqn:Conditional-2a}), we can write
\begin{align}
    \label{eqn:HomogeneousP}
    \bar{P} = E[P] = I + \Delta t E[Q],
\end{align}
and since the sensitive terms in $P$ and $Q$ are assumed to be linear in $x(t)$, we replace $x(t)$ in these terms with $E[x]$ to form $\bar{P}$ and $\bar{Q} := E[Q]$, respectively.
}

Using (\ref{eqn:ReceptorMarkov}) \rev{and} (\ref{eqn:Conditional-2}), (\ref{eqn:MutualInformation})
reduces to
\begin{equation}
    \label{eqn:MutualInfoMarkov}
	I(\mathbf{X};\mathbf{Y}) = \sum_{i=1}^n \sum_{y_i} \sum_{y_{i-1}} \sum_{x_i} p(y_i,x_i,y_{i-1}) \log \frac{p(y_i \given x_i,y_{i-1})}
	{ p(y_i \given y_{i-1})} .
\end{equation}
%
%
%

\rev{
Recall that a transition may be sensitive ($(y_{i-1},y_i) \in \mathcal{S}$) or insensitive ($(y_{i-1},y_i) \not\in \mathcal{S}$). For terms in  (\ref{eqn:MutualInfoMarkov}), consider the
insensitive transitions:
\begin{align}
    \nonumber
	\lefteqn{p(y_i,x_i,y_{i-1}) \log \frac{p(y_i \given x_i,y_{i-1})}
	{ p(y_i \given y_{i-1})}}&\\
	\label{eqn:Insensitive1}
		&= p(y_i,x_i,y_{i-1}) \log \frac{p(y_i \given y_{i-1})}
	{ p(y_i \given y_{i-1})} \\
		&= p(y_i,x_i,y_{i-1}) \log 1\\
		&= 0.
\end{align}
where (\ref{eqn:Insensitive1}) follows since the transition is insensitive, and is not a function of $x_i$; cf. (\ref{eqn:Conditional-2a}).} Thus for IID inputs, the mutual information (\ref{eqn:MutualInfoMarkov}) is calculated using the {\em sensitive transitions only}, i.e., those transitions in $\mathcal{S}$. \rev{With this in mind, we can rewrite (\ref{eqn:MutualInfoMarkov}) as
\begin{align}
	\nonumber \lefteqn{I(\mathbf{X};\mathbf{Y}) }&\\
    \label{eqn:MutualInfoMarkovSensitive0}
	&= \sum_{i=1}^n \sum_{(y_{i-1},y_i)\in\mathcal{S}} \sum_{x_i} p(y_i,x_i,y_{i-1}) \log \frac{p(y_i \given x_i,y_{i-1})}
	{ p(y_i \given y_{i-1})} \\
    \label{eqn:MutualInfoMarkovSensitive}
    &= \sum_{i=1}^n \sum_{\mathcal{A}_i}
    p(y_i \given x_i, y_{i-1}) p(y_{i-1}) p(x_i)
    \log \frac{p(y_i \given x_i,y_{i-1})}
	{ p(y_i \given y_{i-1})} 
\end{align}
where we let $\mathcal{A}_i = \{y_i,y_{i-1},x_i : (y_i,y_{i-1})\in\mathcal{S}^2, x_i \in \mathcal{X}\}$, i.e.
the same terms as the sum in (\ref{eqn:MutualInfoMarkovSensitive}), for the sake of brevity.
Also note that (\ref{eqn:MutualInfoMarkovSensitive}) follows from (\ref{eqn:MutualInfoMarkovSensitive0}) because the input $\mathbf{X}$ is IID.
}

%
%
\rev{Now consider the individual PMFs in (\ref{eqn:MutualInfoMarkovSensitive}), starting with $p(y_i \given x_i, y_{i-1})$.}
All transitions in $\mathcal{S}$ are dependent on the input $x_i$, and throughout this paper we assume that the 
sensitive transition rates depend linearly on the \rev{input signal intensity}. Thus \rev{(recall (\ref{eqn:Markov-last}))}
for non-self-transitions $(y_{i-1},y_i) \in \mathcal{S}$ (i.e., $y_{i-1} \neq y_i$),
\begin{equation}
	\label{eqn:TransitionsInS}
	p(y_i \given x_i,y_{i-1}) = q_{y_{i-1}y_i} x_i \Delta t .
\end{equation}
For self-transitions in $\mathcal{S}$ (i.e., $y_{i-1}\equiv y_i=y$) we have
\begin{align}
    \nonumber\lefteqn{p_{Y_i|X_i,Y_{i-1}}(y \given x_i,y) =}&\\ 
    \label{eqn:SelfTransition}
    & 1-\left(\sum_{y'\not=y,(y',y)\in\mathcal{S}}q_{yy'}x_i    - \sum_{y'\not=y,(y',y)\not\in\mathcal{S}}q_{yy'}\right)\Delta t ,
\end{align}
as seen in the diagonal entries of (\ref{eqn:Markov-last}). 
\rev{Similarly, the terms $p(y_i \given y_{i-1})$ can be obtained 
using (\ref{eqn:Conditional-2a})--(\ref{eqn:HomogeneousP}); we replace $x_i$ 
in (\ref{eqn:TransitionsInS})--(\ref{eqn:SelfTransition}) with $\bar{x}$.}

\rev{The terms $p(y_{i-1})$ represent the steady-state marginal probability that the receptor is in state $y$; for compact notation, let $\pi_{y_{i-1}} = p(y_{i-1})$.}
If the input $x$ is IID, as we assume throughout this paper, then $\pi_{y_{i-1}}$ exists if the Markov chain is irreducible, aperiodic, and positive recurrent; these conditions hold for all the examples we consider \rev{(recall (\ref{eqn:Conditional-2})--(\ref{eqn:HomogeneousP})).}%
\footnote{
\rev{For clarity, although $\pi_y$ may be written with a time-indexing subscript, e.g. $\pi_{y_i}$, this refers to the steady-state distribution of state $y_i \in \mathcal{Y}$, and does not imply that $\pi_{y}$ changes with time.}
}

Define the partial entropy function
\begin{equation}
	\label{eqn:PhiDefinition}
    \phi(p)=
\begin{cases}
0,&p=0\\
p\log p,&p\not=0
\end{cases}
\end{equation}
and let 
\begin{equation}
	\label{eqn:BinEnt}
    \mathscr{H}(p) = -\phi(p) - \phi(1-p)
\end{equation}
represent the binary entropy function.
\rev{Then we have the following result.}

%

\rev{
%
%
\begin{theorem}
    \label{thm:MutualInformationRate}
For an IID input distribution $p(x_i)$, the mutual information rate $\mathcal{I}(X;Y)$ is given by
\begin{align}
	\nonumber \lefteqn{\mathcal{I}(X;Y)} & \\ 
    \nonumber
	&= \sum_{(y_{i-1},y_i) \in \mathcal{S}} \pi_{y_{i-1}}
	\Bigg( 
    \sum_{x_i \in \mathcal{X}} p(x_i) \phi\Big(p(y_i \given x_i,y_{i-1})\Big)\\
	\label{eqn:GeneralIID1}
	&\quad\quad- \phi\left(\sum_{x_i \in \mathcal{X}} p(x_i) p(y_i \given x_i,y_{i-1}) \right) 
	\Bigg). 
\end{align}
Furthermore, $\ciid = \max_{p(x)} \mathcal{I}(X;Y)$.
\end{theorem}
\begin{proof}
    Divide the terms in (\ref{eqn:MutualInfoMarkovSensitive}) into the $i=1$ term, and all the remaining terms.
    Let $T_1(p(x_i))$ represent the $i=1$ term, emphasizing its dependence on the IID input distribution $p(x_i)$, so that
    \begin{align}
        T_1(p(x_i)) &= p(y_1 \given x_1, y_0) p(y_1) p(x_1)
    \log \frac{p(y_1 \given x_1,y_0)}
	{ p(y_1 \given y_0)}\\
	\label{eqn:T1}
	&=p(y_1 \given x_1) p(y_1) p(x_1)
    \log \frac{p(y_1 \given x_1)}
	{ p(y_1)} ,
    \end{align} 
    where (\ref{eqn:T1}) follows since $y_0$ is null. Let $T_2(p(x_i),n)$ represent the remaining terms, again dependent on $p(x_i)$ but also on $n$, so that
    \begin{align}
        \nonumber \lefteqn{T_2(p(x_i),n)}&\\
    &= \sum_{i=2}^n \sum_{\mathcal{A}_i} 
    p(y_i \given x_i, y_{i-1}) p(y_{i-1}) p(x_i)
    \log \frac{p(y_i \given x_i,y_{i-1})}
	{ p(y_i \given y_{i-1})}\\
	\label{eqn:ProofMIRate0}
	&= (n-1) \sum_{\mathcal{A}_i} 
    p(y_i \given x_i, y_{i-1}) p(y_{i-1}) p(x_i)
    \log \frac{p(y_i \given x_i,y_{i-1})}
	{ p(y_i \given y_{i-1})}
    \end{align}
    recalling the definition of $\mathcal{A}_i$ from the discussion after (\ref{eqn:MutualInfoMarkovSensitive}).
    Using (\ref{eqn:InfoRateDefinition}),
    \begin{align}
        \nonumber\lefteqn{\mathcal{I}(X;Y)}&\\
        &= 
        \lim_{n \rightarrow \infty} \frac{T_1(p(x_i))}{n} + \lim_{n \rightarrow \infty} \frac{T_2(p(x_i),n)}{n} \\
        \label{eqn:ProofMIRate1}
	&= \sum_{\mathcal{A}_i} 
    p(y_i \given x_i, y_{i-1}) p(y_{i-1}) p(x_i)
    \log \frac{p(y_i \given x_i,y_{i-1})}
	{ p(y_i \given y_{i-1})} ,
    \end{align}
    and (\ref{eqn:GeneralIID1}) follows after some manipulation.
    
    To show that $\ciid = \max_{p_X(x)} \mathcal{I}(X;Y)$,
    recall the definitions of $\ciid(n)$ and $\ciid$ in (\ref{eqn:IIDCapacityDefinition}) and (\ref{eqn:ciid}), respectively. 
    Referring to $p(x_i)$ as $p$ for brevity,
    \begin{align}
        \ciid(n) = \max_{p} \big(T_1(p) + T_2(p,n)\big) .
    \end{align}
    Let $p_1$ represent the IID input distribution maximizing the term $T_1(p)$, and let $p_2$ represent the IID input distribution maximizing the term $T_2(p,n)$. From (\ref{eqn:ProofMIRate0}), $p_2$ is independent of $n$. Furthermore,
    \begin{align}
        \label{eqn:ProofMIRate2}
        \frac{T_1(p_2)}{n} + \frac{T_2(p_2,n)}{n} \leq \frac{1}{n}\ciid(n) \leq \frac{T_1(p_1)}{n} + \frac{T_2(p_2,n)}{n} .
    \end{align}
    Taking the limit throughout (\ref{eqn:ProofMIRate2}) as $n \rightarrow \infty$, the $T_1$ terms vanish as they are constant with respect to $n$. Comparing (\ref{eqn:ProofMIRate0}) and (\ref{eqn:ProofMIRate1}), $p_2$ also maximizes $\mathcal{I}(X;Y)$. The result follows.
\end{proof}
}


\subsection{Limit of $\mathcal{I}(X;Y)/\Delta t$ as $\Delta t \rightarrow 0$}

In this section we  consider the {\em continuous time limit} of $\mathcal{I}(X;Y) / \Delta t$ as $\Delta t \rightarrow 0$, and give \rev{our second main result} (Theorem \ref{thm:Theorem1}): that in the continuous time limit, the mutual information rate is expressed simply as a product of the average flux through sensitive edges, and the relative entropy between the prior distribution on $x$, and the posterior given a transition. \rev{While we do not claim to derive the mutual information rate of the continuous time channel, the continuous time limit of the discrete-time mutual information rate is a quantity of interest in its own right.}

%
%
\rev{
First, we show that the steady-state distribution $\pi_y$ is independent of $\Delta t$:}
\rev{
\begin{lemma}
\label{lem:SteadyState}
Suppose $\pi_y$ is the normalized left eigenvector of $\bar{Q}$ with eigenvalue 0 (see (\ref{eqn:HomogeneousP})). Define the set $\mathcal{T}$ so that $\Delta t \in \mathcal{T}$ if $P$ from (\ref{eqn:Markov-last}) is a valid transition probability matrix for all $x \in \mathcal{X}$.  Then $\pi_y$ is the normalized left eigenvector of $\bar{P}$ with eigenvalue 1, for all $\Delta t \in \mathcal{T}$.
\end{lemma}
\begin{proof}
    The proof is given in the appendix.
\end{proof}
Note that $\mathcal{T}$ contains all ``sufficiently small'' $\Delta t$. It follows from the lemma that the steady state distribution $\pi_y$ is the same for both continuous and discrete time.
}

Note that the mutual information rate $\mathcal{I}(X;Y)$ in (\ref{eqn:GeneralIID1}) has units of nats per channel use, and that channel uses have duration $\Delta t$. Moreover, the transition probabilities $p(y_i \given x_i,y_{i-1})$ in (\ref{eqn:TransitionsInS})\rev{--}(\ref{eqn:SelfTransition}) are linear functions of $\Delta t$. 
Substituting \rev{the discrete-time transition probabilities \eqref{eqn:Markov-last}} into (\ref{eqn:GeneralIID1}), the non-self-transition probabilities go to zero while the self-transition probabilities go to 1, so 
$\mathcal{I}(X;Y) \rightarrow 0$ as $\Delta t \rightarrow 0$. This should not be surprising: intuitively, as the time step shrinks, less information can be expressed per time step.
However, dividing by $\Delta t$ (and obtaining $\mathcal{I}(X;Y) / \Delta t$),
the information rate {\em per second} is finite. 
It is then useful to consider how this rate behaves as $\Delta t \rightarrow 0$.

Let $\mathcal{S}^\prime \subset \mathcal{S}$ represent the set of sensitive transitions excluding self transitions, i.e.,
\begin{equation}
	\mathcal{S}^\prime = \{(y_{i-1},y_i) :  (y_{i-1},y_i) \in \mathcal{S}, y_{i-1} \neq y_i \} .
\end{equation}
Also let $\mathcal{S}\backslash\mathcal{S}^\prime$ represent the components of $\mathcal{S}$ excluding $\mathcal{S}^\prime$ (i.e., {\em only} the sensitive self transitions).

For any edge $(y,y')$ define the limiting value of that edge's contribution to the mutual information rate, as $\Delta t\to 0$, as
\begin{align}
\nonumber \lefteqn{\iota(y,y') =} & \\
\nonumber
& \lim_{\Delta t\to 0}\frac1{\Delta t}  \pi_{y}
	\Bigg( 
    \sum_{x \in \mathcal{X}} p(x) \phi\Big(p(y' \given x,y)\Big)\\
    &
	\quad\quad - \phi\left(\sum_{x \in \mathcal{X}} p(x) p(y' \given x,y) \right) 
	\Bigg)  
\label{eq:iota}
\end{align}
The limit calculation depends on whether $y=y'$. In case $y\not=y'$, we have $p(y'\given x,y)=q_{yy'}x\Delta t$ (see \eqref{eqn:TransitionsInS}) and
\begin{align}
    &\sum_{x}p(x)\phi(p(y'\given x,y))-\phi\left(\sum_x p(x)p(y'\given x,y) \right) \nonumber \\ \nonumber
    &=\Delta t\Bigg\{ 
    \left(\sum_x q x p(x) \right)\log q +
    \left( 
    \sum_x q  p(x) x\log x 
    \right) \\
    \nonumber
    & \quad\quad - \left(\sum_x q x p(x) \right)\log\left( \sum_x q x p(x) \right)
    \Bigg\} \\
    & \quad\quad +o(\Delta t),\text{ as }\Delta t\to 0^+\\  
    &=  q\Delta t(E(x\log x) - E(x)\log(E(x)))  +o(\Delta t),\text{ as }\Delta t\to 0^+\\ 
    &=q\Delta t(E\phi(x) - \phi(Ex))+o(\Delta t),\text{ as }\Delta t\to 0^+.
\end{align}
On the other hand, in the case when $y=y'$, $\sum_{x}p(x)\phi(p(y'\given x,y))-\phi\left(\sum_x p(x)p(y'\given x,y) \right)=o(\Delta t)$, as $\Delta t\to 0^+$.  Therefore, these terms do not contribute to the mutual information.

Using these \rev{results}, we can rewrite (\ref{eqn:GeneralIID1}) as
\rev{\begin{align}
	\nonumber\lefteqn{\lim_{\Delta t \rightarrow 0}\frac{\mathcal{I}(X;Y)}{\Delta t}}&\\ 
	&= \sum_{(y_{i-1},y_i) \in \mathcal{S}^\prime} \iota(y_{i-1},y_i) 
	+\sum_{(y_{i-1},y_i) \in \mathcal{S} \backslash \mathcal{S}^\prime}\iota(y_{i-1},y_i).
	\label{eqn:GeneralIID2}
\end{align}}
Using (\ref{eqn:TransitionsInS})\rev{--}(\ref{eqn:SelfTransition}), we consider the two additive terms in (\ref{eqn:GeneralIID2}) separately.
For the first term (summing over $\mathcal{S}^\prime$),
we use l'H\^opital's rule: 
in the denominator we have (trivially)
\begin{align}
	\label{eqn:limit-denominator}
	\frac{d}{d\Delta t} \Delta t &= 1,
\end{align}
%
%
and from the numerator, we have
\begin{align}
	\nonumber \lefteqn{
	\lim_{\Delta t \rightarrow 0} \frac{d}{d\Delta t} \sum_{(y_{i-1},y_i) \in \mathcal{S}^\prime} \pi_{y_{i-1}}
	\Bigg( 
    \sum_{x_i \in \mathcal{X}} p(x_i) \phi\Big(q_{y_{i-1}y_i} x_i \Delta t \Big)}\\
    &\quad
	- \phi\left(\sum_{x_i \in \mathcal{X}} p(x_i) q_{y_{i-1}y_i} x_i \Delta t \right) 
	\Bigg)  \\
	\nonumber
	&=  \lim_{\Delta t \rightarrow 0} \sum_{(y_{i-1},y_i) \in \mathcal{S}^\prime} \pi_{y_{i-1}}
	\Bigg( 
    \sum_{x_i \in \mathcal{X}} p(x_i)  \frac{d}{d\Delta t} \phi\Big(q_{y_{i-1}y_i} x_i \Delta t \Big)\\
	&\quad\quad -  \frac{d}{d\Delta t} \phi\left(\sum_{x_i \in \mathcal{X}} p(x_i) q_{y_{i-1}y_i} x_i \Delta t \right) 
	\Bigg)\\
	\nonumber
	&= \sum_{(y_{i-1},y_i) \in \mathcal{S}^\prime} \pi_{y_{i-1}} \Bigg( \sum_{x_i \in \mathcal{X}}p(x_i) q_{y_{i-1}y_i}x_i \log (q_{y_{i-1}y_i}x_i)\\
	&\quad\quad - q_{y_{i-1}y_i} \bar{x} \log (q_{y_{i-1}y_i} \bar{x})
	\Bigg)
	\label{eqn:LimitMutualInfo}
\end{align}
where $\bar{x} = \sum_{x_i \in \mathcal{X}} x_i p(x_i)$ is the average input concentration. 
For the second term (summing over $\mathcal{S}\backslash\mathcal{S}^\prime$), a similar derivation shows that the limit is zero.  

Simplifying further, we have
\begin{align}
	\nonumber\lefteqn{\lim_{\Delta t \rightarrow 0} \frac{\mathcal{I}(X;Y)}{\Delta t}}&\\ 
	\nonumber
	&= \sum_{(y_{i-1},y_i) \in \mathcal{S}^\prime} \pi_{y_{i-1}} \Bigg( \sum_{x_i \in \mathcal{X}}p(x_i) q_{y_{i-1}y_i}x_i \log (q_{y_{i-1}y_i}x_i)\\
	&\quad\quad- q_{y_{i-1}y_i} \bar{x} \log (q_{y_{i-1}y_i} \bar{x}) \Bigg) \\
	\nonumber
	&= \sum_{(y_{i-1},y_i) \in \mathcal{S}^\prime} \pi_{y_{i-1}} q_{y_{i-1}y_i} 
	 \bar{x}\log (q_{y_{i-1}y_i}) \\
	 \nonumber
	 &\quad\quad+ 
	 \sum_{(y_{i-1},y_i) \in \mathcal{S}^\prime} \pi_{y_{i-1}} q_{y_{i-1}y_i}\sum_{x_i \in \mathcal{X}}p(x_i) x_i\log(x_i) \\
	 \nonumber
	&\quad\quad- \sum_{(y_{i-1},y_i) \in \mathcal{S}^\prime} \pi_{y_{i-1}} q_{y_{i-1}y_i} \bar{x} \log (q_{y_{i-1}y_i})\\
	&\quad\quad- \sum_{(y_{i-1},y_i) \in \mathcal{S}^\prime} \pi_{y_{i-1}} q_{y_{i-1}y_i} \bar{x} \log(\bar{x}) \\
	\nonumber
	&= \sum_{(y_{i-1},y_i) \in \mathcal{S}^\prime} \pi_{y_{i-1}} q_{y_{i-1}y_i}\sum_{x_i \in \mathcal{X}}p(x_i) x_i\log(x_i)\\
	&\quad\quad
	- \sum_{(y_{i-1},y_i) \in \mathcal{S}^\prime} \pi_{y_{i-1}} q_{y_{i-1}y_i} \bar{x} \log(\bar{x}) \\
	\nonumber
	&= \left( \sum_{(y_{i-1},y_i) \in \mathcal{S}^\prime} \pi_{y_{i-1}} q_{y_{i-1}y_i} \right)\\
	&\quad\quad \cdot
	\left( \sum_{x_i \in \mathcal{X}}p(x_i) x_i\log(x_i) - \bar{x}\log\bar{x} \right) .
	\label{eqn:FactoredMutualInformation}
\end{align}

The {\em steady-state flux} $J_{y_{i-1}y_i}$ through an edge $(y_{i-1},y_i)$ in the state transition graph 
is defined as 
\begin{equation}
	J_{y_{i-1}y_i} := \pi_{y_{i-1}} q_{y_{i-1}y_i} \bar{x} .
\end{equation}
Similarly, the {\em net steady-state flux} through the sensitive (non-self) edges in the graph is
\begin{align}
	J_{\mathcal{S}^\prime} &:= \sum_{(y_{i-1},y_i) \in \mathcal{S}^\prime} J_{y_{i-1}y_i} \\
	\label{eqn:SteadyStateFlux}
	&= \sum_{(y_{i-1},y_i) \in \mathcal{S}^\prime} \pi_{y_{i-1}} q_{y_{i-1}y_i} \bar{x} ,
\end{align}
Expressing (\ref{eqn:FactoredMutualInformation}) in terms of $J_{\mathcal{S}^\prime}$,
we have
\begin{align}
	\lim_{\Delta t \rightarrow 0} \frac{\mathcal{I}(X;Y)}{\Delta t}
	&= \frac{1}{\bar{x}} J_{\mathcal{S}^\prime} 
	\left( \sum_{x_i \in \mathcal{X}}p(x_i) x_i\log x_i - \bar{x}\log\bar{x} \right) \\
	\label{eqn:FactoredMutualInformation2}
	&= J_{\mathcal{S}^\prime} 
	\left( \sum_{x_i \in \mathcal{X}} \frac{p(x_i)x_i}{\bar{x}}\log x_i - \log\bar{x} \right)
\end{align}
\rev{We define 
\begin{equation}
	\label{eqn:nu}
	\nu(x_i) := \frac{p(x_i) x_i}{\bar{x}} . 
\end{equation}
}
Since $\nu(x_i)$ is positive for all $x_i$, and since
\begin{align}
	\sum_i \nu(x_i) &= \sum_i \frac{p(x_i) x_i}{\bar{x}}  = \frac{\bar{x}}{\bar{x}} = 1,
\end{align}
\rev{it follows that} $\nu(x_i)$ forms a probability distribution, in general different from $p(x_i)$.
We discuss the physical interpretation of $J_{\mathcal{S}^\prime}$ and 
$\nu(x_i)$ in the next section.

Using $\nu(x_i)$, we can rewrite (\ref{eqn:FactoredMutualInformation2}) as
\begin{align}
	\lim_{\Delta t \rightarrow 0} \frac{\mathcal{I}(X;Y)}{\Delta t}
	&= J_{\mathcal{S}^\prime} 
	\left( \sum_{x_i \in \mathcal{X}} \nu(x_i) \log x_i - \log\bar{x} \right) \\
	&= J_{\mathcal{S}^\prime} 
	\left( \sum_{x_i \in \mathcal{X}} \nu(x_i) \log x_i -  \sum_{x_i \in \mathcal{X}} \nu(x_i) \log\bar{x} \right) \\
	&= J_{\mathcal{S}^\prime} 
	\left( \sum_{x_i \in \mathcal{X}} \nu(x_i) \log \frac{x_i}{\bar{x}} \right) \\
	&= J_{\mathcal{S}^\prime} 
	\left( \sum_{x_i \in \mathcal{X}} \nu(x_i) \log \frac{p(x_i)x_i}{p(x_i)\bar{x}} \right)\\
	&= J_{\mathcal{S}^\prime} 
	\left( \sum_{x_i \in \mathcal{X}} \nu(x_i) \log \frac{\nu(x_i)}{p(x_i)} \right)\\
	&= J_{\mathcal{S}^\prime} D(\nu \:\rev{\|}\: p) ,
\end{align}
where $D(\cdot\:\rev{\|}\:\cdot)$ represents the Kullback-Leibler divergence.

The preceding derivation, \rev{including Lemma \ref{lem:SteadyState},} allows us to state the following result.
\begin{theorem}
	\label{thm:Theorem1}
	For finite-state Markov signal transduction systems described by (\ref{eqn:Markov-last}), with IID inputs,
	\begin{equation}
    	\label{eqn:Theorem1}
		\lim_{\Delta t \rightarrow 0} \frac{\mathcal{I}(X;Y)}{\Delta t} = J_{\mathcal{S}^\prime} D(\nu \:\rev{\|}\: p) ,
	\end{equation}
	with $J_{\mathcal{S}^\prime}$ defined in (\ref{eqn:SteadyStateFlux}) and $\nu$ defined in (\ref{eqn:nu}).
\end{theorem}

\subsection{Physical interpretation}

The factorization in (\ref{eqn:FactoredMutualInformation}) gives us a useful physical interpretation of 
the mutual information in this system. 

First consider $J_{\mathcal{S}^\prime}$. Physically, if one watched only for transitions along 
edge $(y_{i-1},y_i)$ (with the rest of the graph assumed to be at steady state), 
$J_{y_{i-1}y_i}$ gives 
the average rate at which those transitions would be observed; that is, $J_{y_{i-1}y_i}$ is the \pt{mean} flux through the transition $y_{i-1} \rightarrow y_i$. Thus, $J_{\mathcal{S}^\prime}$ is the average rate through {\em all} the sensitive edges, i.e., the {\em net flux}.

Now consider $D(\nu \:\rev{\|}\: p)$, and note that the distribution $\nu(x_i)$ is a {\em posterior} distribution of $x_i$. To see this,
consider a random variable $y \in \{0,1\}$, with conditional distribution
\begin{equation}
	\label{eqn:YXconditional}
	p_{Y|X}(1 \given x_i) =\kappa x_i ,
\end{equation}
where $\kappa$ is a positive constant ($0 \leq \kappa \leq \frac{1}{x_i}$ to make a valid probability),
and $p_{Y|X}(0 \given x_i) = 1- p_{Y|X}(1 \given x_i)$.
The marginal distribution $p_Y(1)$ is then given by
\begin{align}
	p_Y(1) &= \sum_{x_i} p(x_i) p_{Y|X}(1 \given x_i)\\
	&= \sum_{x_i} p(x_i) \kappa x_i\\
	&= \kappa \bar{x} .
\end{align}
With this definition, $\nu(x_i)$ is the posterior distribution of $x$ given $y=1$:
\begin{align}
	p_{X|Y}(x_i \given 1) &= \frac{p(x_i)p_{Y|X}(1 \given x_i)}{p_Y(1)}\\
	&= \frac{p(x_i)  \kappa x_i}{\kappa \bar{x}}\\
	&= \frac{p(x_i) x_i}{\bar{x}} = \nu(x_i) .
\end{align}
Physically, consider the example of a ligand-gated channel where $x_i$ is the concentration of ligands near the receptor at input $i$. With $i \in \{\L,\H\}$ (i.e., inputs $x_\L$ and $x_\H$), suppose we select one molecule at random from those near the receptor, and set $y = 1$ if the molecule is a ligand; $y = 0$ otherwise. Then $p_{Y|X}(1 \given x_\L) \propto x_L$ and $p_{Y|X}(1 \given x_\H) \propto x_\H$, with $\kappa$ as the constant of proportionality; this satisfies (\ref{eqn:YXconditional}). For example, suppose $x_\L$ is measured in {\em number concentration} of ligands, i.e., number of ligands per volume $V$. Then $p_{Y|X}(1 \given x_i) = x_i / n$ (for $i \in \{\L,\H\}$), where $n$ is the number concentration of all molecules, ligands and otherwise, near the receptor, and $\kappa = 1/n$.

In general, physical systems where the probability of response $p(y \given x)$ is directly proportional to the input $x$ 
fit into this framework, emphasizing the importance of this modeling assumption made in Section II.

\subsection{Shannon capacity of receptors with a single sensitive non-self transition}

We now give our \rev{third} main result, showing that the Shannon capacity $C$ is equal to the IID capacity $\ciid$ for a number of sensitive transitions $|\mathcal{S}^\prime| \leq 1$, and furthermore that the capacity-achieving distribution has a simple form. As a consequence, this leads directly to the Shannon capacity of ChR2; we give this capacity in the example below. The result is a generalization of related results in \cite{ThomasEckford2016}.

Recall $\mathcal{S}^\prime \subset \mathcal{S}$ represent the set of transitions, {\em excluding} self-transitions.

\begin{theorem}
\label{thm:capacity}
For any receptor with $|\mathcal{S}^\prime| \leq 1$,
\begin{enumerate}
	\item $\ciid$ is achieved with all probability mass on $x_\L$ and $x_\H$; and
    \item $C = \ciid$.
\end{enumerate}
\end{theorem}
\begin{proof}
	The case of $|\mathcal{S}^\prime| = 0$ is trivial: the state is never sensitive to the input, so $\mathcal{I}(X;Y) = 0$ for all input distributions. 
    
	Now consider $|\mathcal{S}^\prime| = 1$. We sketch the proof: results in \cite{ThomasEckford2016} were presented for a two-state receptor where only one transition was sensitive; many of the results have the same form.
The first part of the theorem follows from \cite[Thm 1]{ThomasEckford2016}, noting from (\ref{eqn:GeneralIID1}) that any system with $|\mathcal{S}^\prime| = 1$ has the same form, apart from the marginal distribution $\pi_{y_{i-1}}$, which is held constant in the proof of \cite[Thm 1]{ThomasEckford2016}. The second part of the theorem follows from \cite[Thm 2]{ThomasEckford2016}, noting that $C$ is only a function of the input distribution in the sensitive state. 
\end{proof}

\subsection{Example}

We now give an example calculation of the mutual information and IID capacity, \pt{by which we obtain the channel capacity of channelrhodopsin.}
\begin{ex}[ChR2]
	\label{ex:ChR2-MI}
	Referring to the rate matrix for ChR2 (\ref{eqn:ChR2-rate-matrix}),
	there are exactly two sensitive transitions: first, the 
	transition from $\C_1$ to $\O_2$, represented
	by $q_{12} x(t)$; and second, the self-transition from $\C_1$ to $\C_1$, represented by 
	$R_1 = - q_{12}x(t)$. Thus, $\mathcal{S} = \{(\C_1,\O_2), (\C_1,\C_1)\}$
    and $\mathcal{S}^\prime = \{(\C_1,\O_2)\}$.
	
Suppose $\mathcal{X} = \{x_\L, x_\H\}$, i.e., the input light source can only be off ($x_\L$) or on ($x_\H$). Let $p_\L = \Pr(x = x_\L)$ and $p_\H = \Pr(x = x_\H) = 1-p_\L$.
    
	Recalling the transformation of
	rates into probabilities (\ref{eqn:Markov-last}), and substituting into (\ref{eqn:GeneralIID1}), we have
	\begin{align}
		\nonumber\lefteqn{\mathcal{I}(X;Y)} &\\ 
		\nonumber &= 
		\pi_{\C_1} \Big( p_\L \phi(\Delta t q_{12}x_\L) + p_\H \phi(\Delta t q_{12}x_\H)\\
		\nonumber
		&\quad\quad- \phi \big( p_\L \Delta t q_{12}x_\L + p_\H \Delta t q_{12}x_\H \big) \Big)\\
		\nonumber
        & \quad + \: \pi_{\C_1} \Big( p_\L \phi(1-\Delta t q_{12}x_\L) + p_\H \phi(1- \Delta t q_{12}x_\H)\\
		&\quad\quad - \phi \big( 1-p_\L \Delta t q_{12}x_\L - p_\H \Delta t q_{12}x_\H \big) \Big)
	\end{align}
    where the first term represents the transition $(\C_1,\O_2)$, and the second term represents the self-transition $(\C_1,\C_1)$, both of which are sensitive.
    Continuing the derivation,
    \begin{align}
		\nonumber\lefteqn{\mathcal{I}(X;Y)} &\\
		\nonumber
		&= \pi_{\C_1} \Big( 
        \binent(p_\L \Delta t q_{12}x_\L + p_\H \Delta t q_{12}x_\H)\\
        &\quad\quad - p_\L \binent(\Delta t q_{12}x_\L) - p_\H \binent(\Delta t q_{12}x_\H)
			 \Big) \\
       \label{eqn:ChR2MutualInformationExample}
		&= \left( \frac{q_{23}q_{31}}{q_{23}q_{31} + \bar{x}q_{12}q_{31} + \bar{x}q_{12}q_{23}} \right)\\
		&\quad\quad\cdot
		\Big( \binent(\Delta t q_{12} \bar{x}) - p_\L \binent(\Delta t q_{12}x_\L) - p_\H \binent(\Delta t q_{12}x_\H)
 \Big) ,
	\end{align}
where $\bar{x}$ is the average input. 

Finally, consider $\mathcal{I}(X;Y) / \Delta t$ as $\Delta t \rightarrow 0$, as in (\ref{eqn:LimitMutualInfo}).
The steady-state occupancy probability of $\C_1$, $\pi_{\C_1}$, is independent of $\Delta t$. Thus, 
from Theorem \ref{thm:Theorem1}, we have
\begin{align}
	\nonumber\lefteqn{\lim_{\Delta t \rightarrow 0} \frac{\mathcal{I}(X;Y)}{\Delta t}}&\\ 
	&= J_{\mathcal{S}^\prime} D(\nu \:\rev{\|}\: p)\\
	&= 
	\left( \sum_{(y_{i-1},y_i) \in \mathcal{S}^\prime} \pi_{y_{i-1}} q_{y_{i-1}y_i} \right)
	\left( \sum_{x_i \in \mathcal{X}}\nu(x_i) \log\frac{\nu(x_i)}{p(x_i)} \right) \\
	&= 
	\pi_{\C_1} q_{12}
	\left(  p_\L\frac{x_\L}{\bar{x}} \log \frac{x_\L}{\bar{x}} + p_\H\frac{x_\H}{\bar{x}} \log \frac{x_\H}{\bar{x}} \right) \\
	\nonumber
	&= \frac{q_{12}q_{23}q_{31}}{q_{23}q_{31} + \bar{x}q_{12}q_{31} + \bar{x}q_{12}q_{23}}\\
	&\quad\quad\cdot
	\left(  p_\L\frac{x_\L}{\bar{x}} \log \frac{x_\L}{\bar{x}} + p_\H\frac{x_\H}{\bar{x}} \log \frac{x_\H}{\bar{x}} \right) . 
	\label{eqn:ChR2MutualInformationExample2}
\end{align}

In Figure \ref{fig:MutualInformationFigure}, we illustrate the effect of step size on the mutual information calculation, using  (\ref{eqn:ChR2MutualInformationExample}) for the solid lines (for various values of $\Delta t > 0$) and (\ref{eqn:ChR2MutualInformationExample2}) for the dotted line (as $\Delta t \rightarrow 0$).  From this figure, the IID capacity and the capacity-achieving value of $p_\L$ may be found by taking the maximum over the curve of interest. This value clearly changes for different values of $\Delta t$; however, the IID capacity is around $\ciid \approx 65$ bits/s, and the capacity-achieving $p_\L$ is around $p_\L \approx 0.99$. For any finite value of $\Delta t$, it is interesting to note that the discrete-time approximation over-estimates the mutual information as $\Delta t \rightarrow 0$. 
%
\end{ex}

\begin{figure}[t!]
\begin{center}
\includegraphics[width=3.5in]{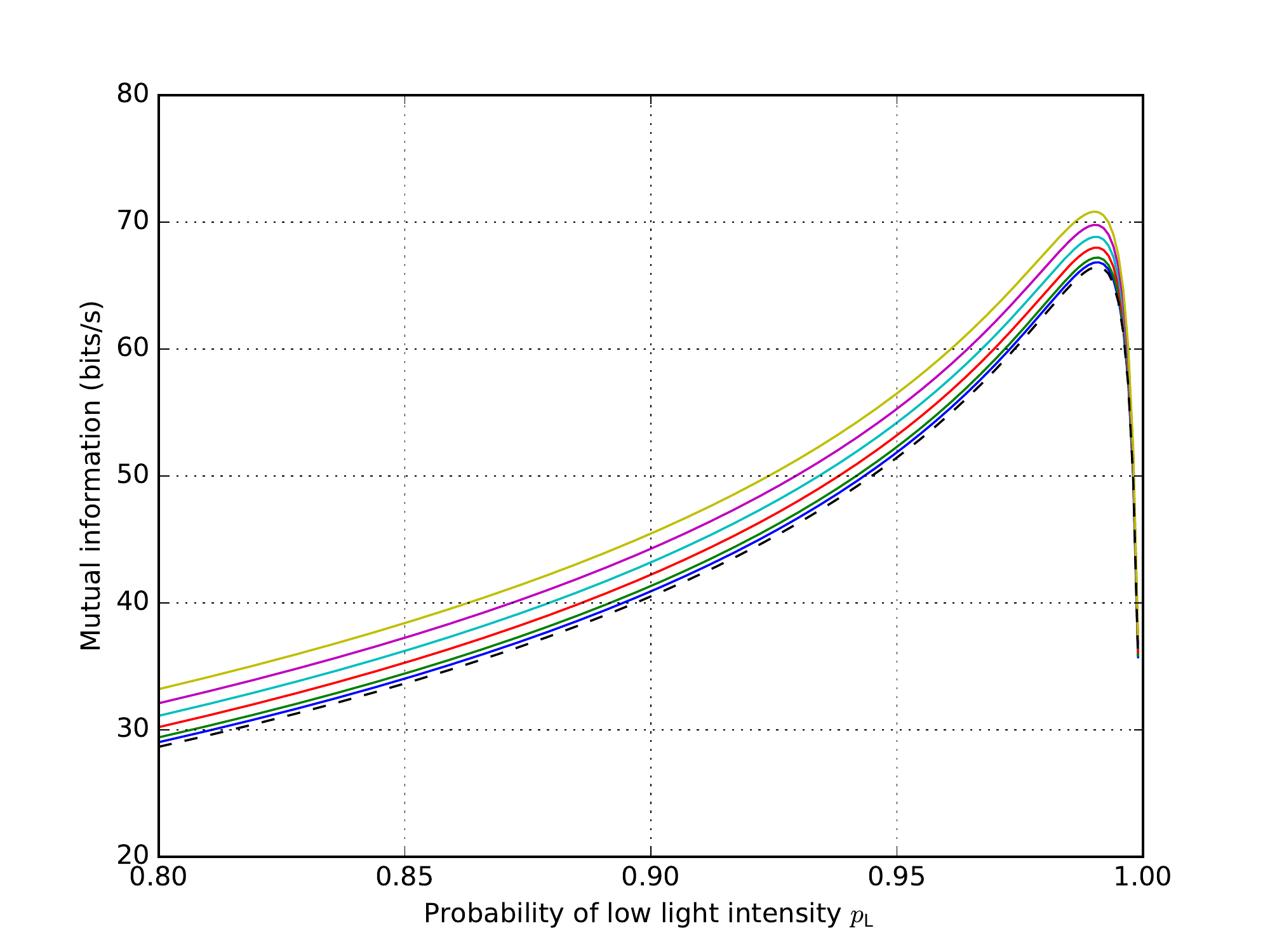}
\end{center}
\caption{\label{fig:MutualInformationFigure} Plot for ChR2, illustrating the effect of $\Delta t$ on $\mathcal{I}(X;Y)$ from (\ref{eqn:ChR2MutualInformationExample}). The dashed black line represents $\Delta t \rightarrow 0$. Solid lines, from bottom, represent: $\Delta t = 0.01$ (blue), $\Delta t = 0.02$ (green), $\Delta t = 0.04$ (red), $\Delta t = 0.06$ (cyan), $\Delta t = 0.08$ (magenta), and $\Delta t = 0.1$ (tan), all in milliseconds.}
\end{figure}

From Example \ref{ex:ChR2-MI}, ChR2 has $|\mathcal{S}^\prime| = 1$. Thus, {\em ChR2 satisfies the conditions of Theorem \ref{thm:capacity}}, and has $C = \ciid$, where $\ciid$ is given in (\ref{eqn:ChR2MutualInformationExample2}).  
\rev{Performing the maximization numerically, on the $\Delta t \rightarrow 0$ line, the maximum value of $\mathcal{I}(X;Y)$ is found near $p_\L = 0.99$ where $\mathcal{I}(X;Y) = 66$ bits/s, which gives} \pt{the channel capacity \rev{$C$} (\textit{sensu} Shannon) of channelrhodopsin.}

A similar calculation can be performed for ACh and CaM. However, the resulting expressions are not as compact as (\ref{eqn:ChR2MutualInformationExample2}), so we exclude them from the paper. Mutual information plots for ACh and CaM (from which $\ciid$ may \rev{also be obtained numerically}) are given in Figures \ref{fig:AChFigure} and \ref{fig:CaMFigure}, respectively.
However, ACh and CaM both have $|\mathcal{S}^\prime| > 1$ (see Figures \ref{fig:ACh} and \ref{fig:CaM}), and do not satisfy the condition in Theorem \ref{thm:capacity}. It remains an open question as to whether $C = \ciid$ for these receptors. The proof of \cite[Thm 2]{ThomasEckford2016} (and of Theorem \ref{thm:capacity}) relies on the feedback capacity being achieved by the IID input distribution. However, if there is more than one sensitive transition, the receiver can use the feedback to distinguish between these transitions, and can select an optimal input distribution for each. Thus, the feedback-capacity-achieving input distribution depends on the feedback, and is not \rev{necessarily} IID. If $C = \ciid$, a different proof technique is required, and we do not address this case.

\section{Discussion}
In this paper we have presented a general framework for signal transduction systems, in which the states of a receptor form a directed graph, some subset of the edges of which represent transitions with intensities modulated by an external signal.  This signal provides the channel input, and the state of the receptor -- a trajectory on the graph -- represents the channel output.  
\pt{We illustrate the signal transduction model, the calculation of mutual information and the IID capacity for several examples: light intensity transduction by channel rhodopsin, acetylcholine concentration transduction by the nicotinic acetylcholine receptor, and transduction of intracellular calcium ion concentration by the calmodulin protein.}

Several caveats are in order, which qualify our results and motivate our future work.

In many signal transduction systems, only a subset of the receptor states engender an observable output signal.  For example, the channelrhodopsin receptor states $\C_1,\O_2,\C_3$ (\textit{cf}.~Fig.~\ref{fig:ChR2}) are not directly observed by the cell in the membrane of which the receptor is embedded; rather it is the net current (zero for states $\C_1,\C_3$ and finite for state $\O_2$) that impacts the rest of the cell.  Similarly, for the nicotinic acetylcholine receptor (\textit{cf}.~Fig.~\ref{fig:ACh}) the state of the receptor as observed by the cell is either ``open'' (states $\O_1,\O_2$) or ``closed'' (states $\C_3,\C_4,\C_5$).  For the calmodulin receptor, there are understood to be four functionally distinct states: both occupied \Ca~binding sites on the N-terminus end of the protein, both occupied \Ca~binding sites on the C-terminus end of the protein, all four \Ca~binding sites occupied, or fewer than two on each end (\textit{cf.}~Fig.~\ref{fig:CaM}; dashed lines indicate physiologically equivalent states).
The diagram \eqref{diag:xyz} shows the general structure of such a channel, with output $Z(t)$ a function of the channel state $Z=f(Y(t))$ (compare with the diagram in (\ref{diag:xy})):
\begin{equation}\label{diag:xyz}
\begin{array}{ccccccccccc}
&X_1&&X_2&&X_3&&X_4&&X_5&\cdots\\
&\downarrow&&\downarrow&&\downarrow&&\downarrow&& \downarrow   \\
(Y_0)&\longrightarrow&Y_1&\longrightarrow&Y_2&\longrightarrow&Y_3&\longrightarrow&Y_4&\cdots\\
&&\downarrow&&\downarrow&&\downarrow&& \downarrow   \\
&&Z_1&&Z_2&&Z_3&&Z_4&\cdots
\end{array}\end{equation}
By virtue of the information processing inequality, the mutual information rate between $X$ and $Z$ cannot exceed that between $X$ and $Y$.  Preliminary results suggest that the size of the difference -- the information gap --  depends strongly on the network architecture, and the positioning of sensitive edges relative to observable transitions (data not shown).  Detailed consideration of mutual information for Markovian signal transduction channels with such partially observed outputs will be undertaken elsewhere.

We have assumed that the directed edges comprising the receptor's state transition graph fall into two classes, either insensitive (fixed transition rates) or sensitive (transition rates proportional to the input signal intensity).  A more realistic assumption would allow for a dark current (finite transition rate at zero signal intensity), a nonlinear, monotonically increasing transition rate as a function of increasing intensity, or a signaling threshold or minimum intensity value.  Under the IID  input scenario it is optimal to limit the input values to those inducing the maximal and minimal transition rates, in which case several more realistic scenarios could in principle be reduced to the scenario we consider here. For example, a dark current could be captured by adding an additional insensitive channel parallel to a sensitive channel.

We have considered a general class of signal transduction models that are naturally framed as continuous time channels.  Our basic signal transduction channel model process is conditionally Markovian, given the (time varying) input signal.  The simplest model in this class would correspond to Kabanov's Poisson channel \cite{Kabanov1978}, consisting of a single transition with rate modulated by the input.  In order to simplify the analysis of such models it is convenient to translate them into analogous discrete-time models.  The general structure of such as model is a finite state, discrete-time channel in which the probability transition matrix is modulated by the (discrete time) input sequence.  Our previously-discussed results \cite{ThomasEckford2016} introduced a minimal such model, the BIND channel, consisting of a single receptor molecule with two states (bound, $\B$, and unbound, $\U$) with one transition rate ($\U\to\B$) sensitive to the input (ligand molecule concentration) and the other transition rate ($\B\to\U$) insensitive.  In general, the structure of a conditionally Markovian signal-transduction channel under time discretization corresponds to the Unit Output Memory (UOM) channel class analyzed by \rev{Chen and Berger} \cite{che05}.  As mentioned previously,
in \rev{\cite{AsnaniPermuterWeissman2013IEEE_ISIT,PermuterAsnaniWeissman2014IEEETransIT} Asnani, Permuter and Weissman} present several examples of UOM channels that they call POST (prior output is the state) channels, which are also special cases of the channels analyzed by Chen and Berger.  (The BIND channel can be interpreted as a type of POST channel although it is distinct from the examples in \cite{AsnaniPermuterWeissman2013IEEE_ISIT,PermuterAsnaniWeissman2014IEEETransIT}.)
Thus our channel models for channel rhodopsin, the nicotinic acetylcholine receptor and calmodulin may all be seen as examples of Chen and Berger's UOM channel class.


\begin{figure}[t!]
\begin{center}
\includegraphics[width=3.5in]{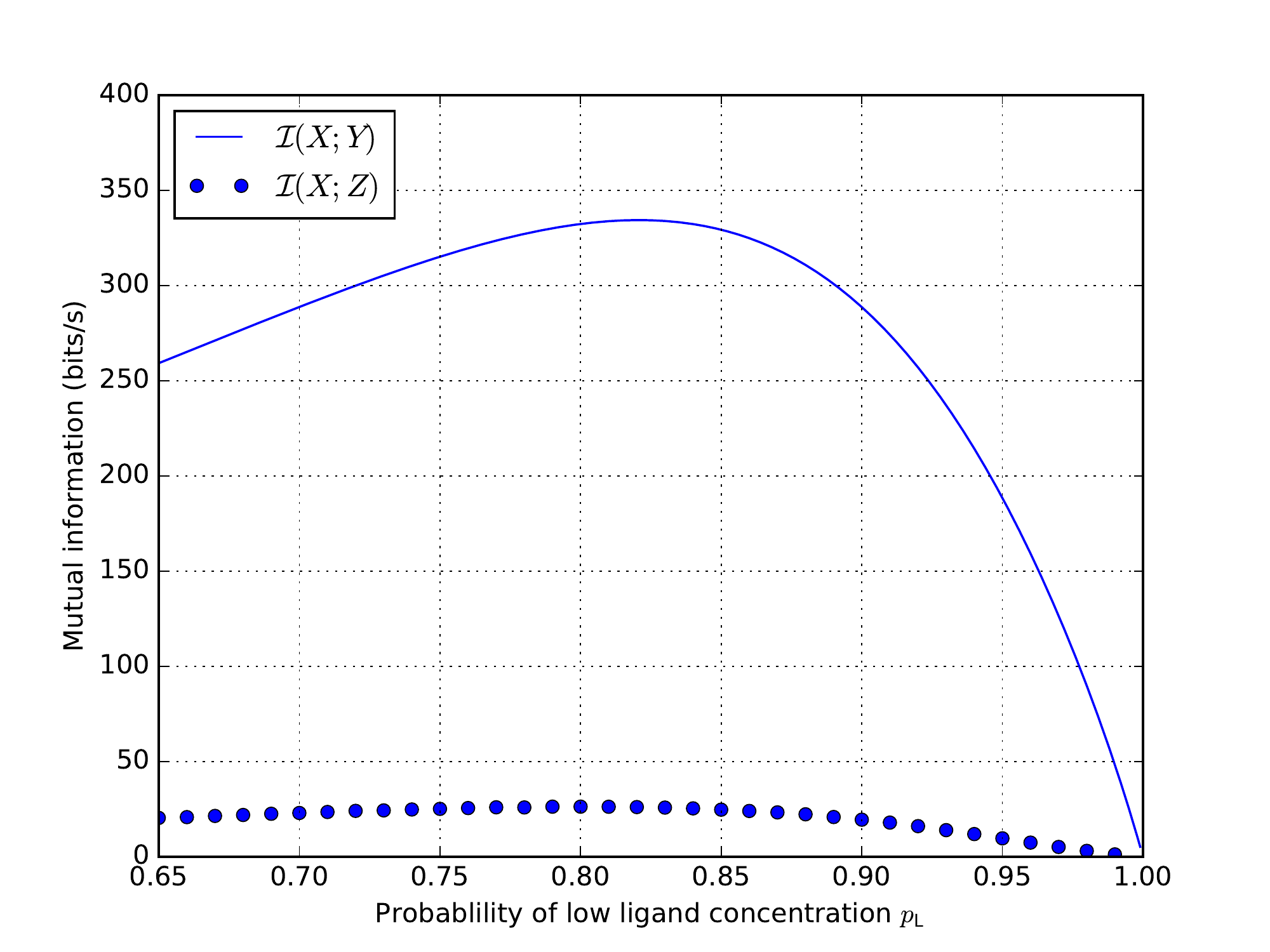}
\end{center}
\caption{\label{fig:AChFigure} Plots of $\mathcal{I}(X;Y)$ and $\mathcal{I}(X;Z)$ for ACh, using $\Delta t = 0.02$ ms. Solid line is from (\ref{eqn:GeneralIID1}), while dots represent {\em Monte Carlo} simulations.}
\end{figure}

\begin{figure}[t!]
\begin{center}
\includegraphics[width=3.5in]{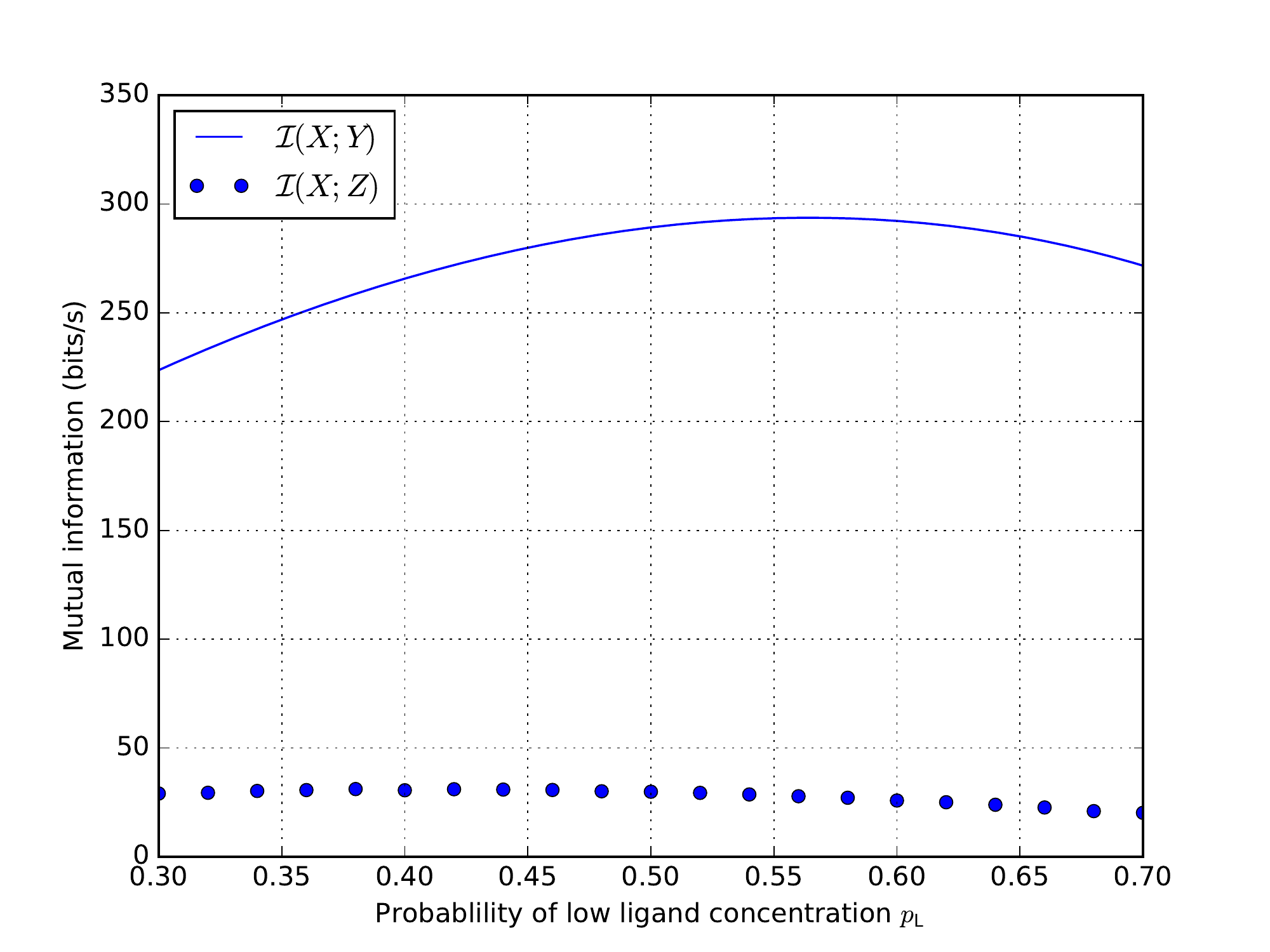}
\end{center}
\caption{\label{fig:CaMFigure} Plots of $\mathcal{I}(X;Y)$ and $\mathcal{I}(X;Z)$ for CaM, using $\Delta t = 0.002$ ms. Solid line is from (\ref{eqn:GeneralIID1}), while dots represent {\em Monte Carlo} simulations.}
\end{figure}

\appendix

%
%
\rev{

By definition (see (\ref{eqn:HomogeneousP}))
\begin{align}
    \pi_y \bar{Q} = 0 \cdot \pi_y = 0 ,
\end{align}
so
\begin{align}
    \pi_y \Delta t \bar{Q} = 0 \cdot \pi_y = 0 ,
\end{align}
i.e. $\pi_y$ is also the zero eigenvector of $\Delta t \bar{Q}$, for any $\Delta t$.

It is well known that adding $I$ to a matrix adds 1 to each eigenvalue. 
Thus, 
\begin{align}
    \label{eqn:Lem1Proof1}
    \pi_y \bar{P} = \pi_y (I + \Delta t \bar{Q}) &= (0 + 1) \pi_y = \pi_y .
\end{align}
Uniqueness of $\pi_y$ follows from the Perron-Frobenius theorem, since $\bar{Q}$ is irreducible, so the lemma follows from (\ref{eqn:Lem1Proof1}).
}

\bibliographystyle{IEEEtran} 
\bibliography{MolecularInfoTheory,infotheory,signaling,Cowan,PJT,neuroscience,Dicty,stoch_chem,MCell}

\begin{IEEEbiography}[{\includegraphics[width=1in,height
=1.25in,clip,keepaspectratio]{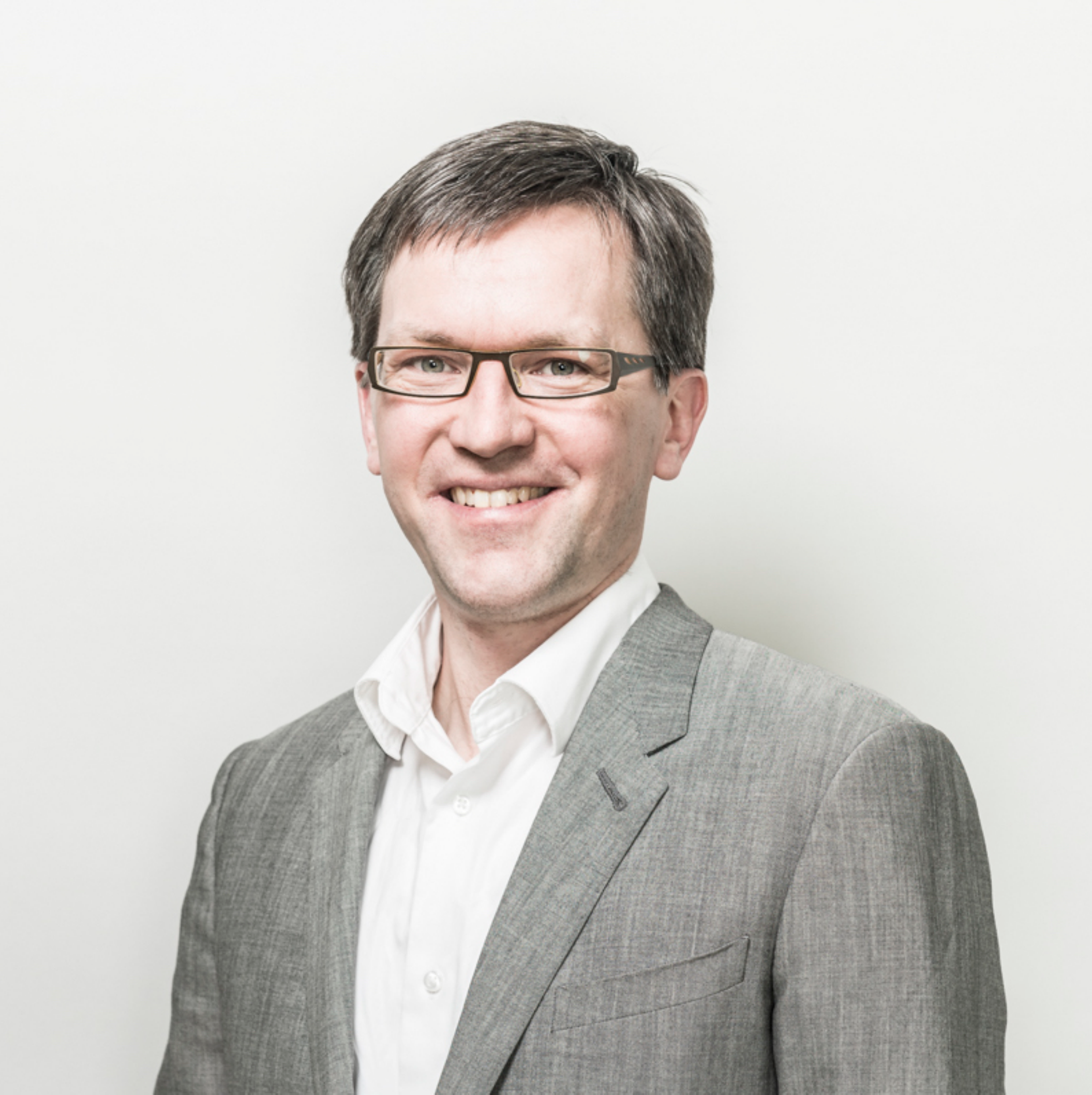}}]{Andrew W. Eckford}
is an Associate Professor in the Department of Electrical Engineering and Computer Science at York University, Toronto, Ontario. He received the B.Eng. degree from the Royal Military College of Canada in 1996, and the M.A.Sc. and Ph.D. degrees from the University of Toronto in 1999 and 2004, respectively, all in Electrical Engineering. Andrew held postdoctoral fellowships at the University of Notre Dame and the University of Toronto, prior to taking up a faculty position at York in 2006. Andrew's research interests include the application of information theory to nonconventional channels and systems, especially the use of molecular and biological means to communicate. His research has been covered in media including {\em The Economist}, {\em The Wall Street Journal}, and {\em IEEE Spectrum}. 
His research received the 2015 IET Communications Innovation Award, and was a finalist for the 2014 Bell Labs Prize. Andrew is also a co-author of the textbook Molecular Communication, published by Cambridge University Press.
\end{IEEEbiography}

\begin{IEEEbiography}[{\includegraphics[width=1in,height
=1.25in,clip,keepaspectratio]{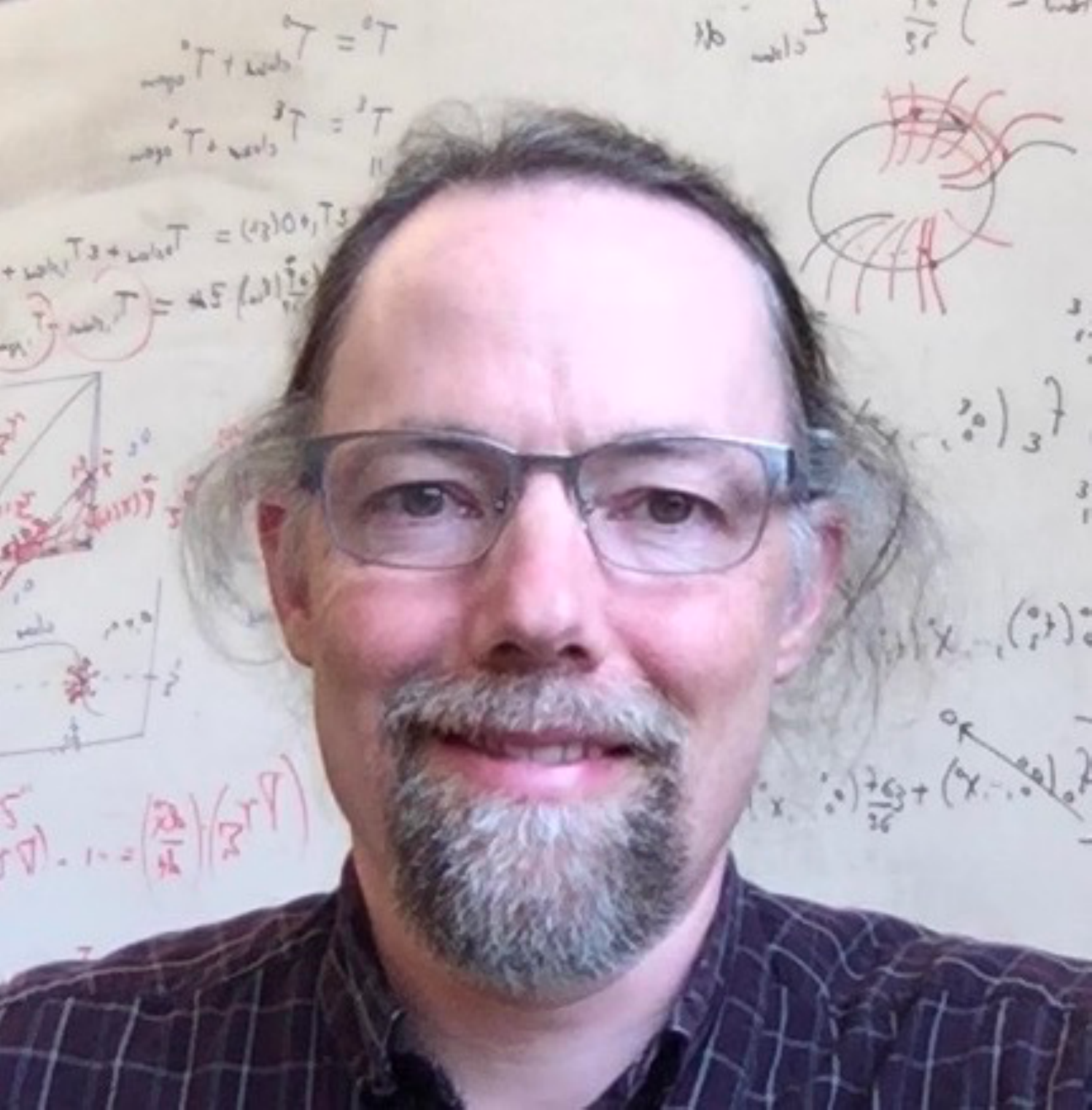}}]{Peter J. Thomas} studies communication and control in complex adaptive biological systems. He obtained an M.A. in the Conceptual Foundations of Science and a Ph.D. in Mathematics from The University of Chicago in 2000. Following postdoctoral work in the Computational Neurobiology Laboratory at The Salk Institute for Biological Studies, he taught mathematics, neuroscience, and computational biology first at Oberlin College and now at Case Western Reserve University, where he is Professor in the Department of Mathematics, Applied Mathematics, and Statistics. He holds secondary appointments in Biology, in Cognitive Science, and in Electrical Engineering and Computer Science.  Prof. Thomas has been a Fulbright scholar and a recipient of a Simons Foundation fellowship.  He currently serves as co-Editor-in-Chief of {\em Biological Cybernetics}.
\end{IEEEbiography}

\end{document}